\documentclass[a4paper,UKenglish,cleveref,autoref,thm-restate]{lipics-v2021}
\usepackage[utf8]{inputenc}
\usepackage{xspace}
\usepackage{ifthen}

\newboolean{shortver}
\setboolean{shortver}{false}
\usepackage{booktabs}
\usepackage{tabularx}
\usepackage{placeins}
\usepackage{needspace}
\usepackage{tikz}
\usepackage{xcolor}
\nolinenumbers

\newcommand{\pname}[1]{\textnormal{\textsc{#1}}}
\newcommand{\NPC}{NP-complete}
\newcommand{\HFD}{\pname{$H$-free Edge Deletion}}
\newcommand{\HFS}{\pname{$H$-free Subdivision}}
\newcommand{\girth}{\operatorname{girth}}
\newcommand{\dist}{\operatorname{dist}}
\newcommand{\core}{\operatorname{core}}

\title{A Complexity Dichotomy for Edge Subdivision to $H$-Free Graphs}
\titlerunning{Complexity Dichotomy for $H$-Free Subdivision}
\author{Marta Piecyk}
{CISPA Helmholtz Center for Information Security, Germany}
{marta.piecyk@cispa.de}
{0000-0001-9162-8300}
{}
\author{R. B. Sandeep}
{Indian Institute of Technology Dharwad, India}
{sandeeprb@iitdh.ac.in}
{0000-0003-4383-1819}
{Supported by Anusandhan National Research Foundation (India) grant MTR/2022/000692. Gratefully acknowledges the financial support provided by the Faculty of Mathematics, Informatics, and Mechanics, University of Warsaw, for his research visit to the University of Warsaw, Poland.}
\authorrunning{M. Piecyk and R.\,B. Sandeep}
\Copyright{Marta Piecyk and R. B. Sandeep}
\hideLIPIcs

\acknowledgements{The project was initiated at STRUG: Structural Graph Theory Bootcamp, funded by the ``Excellence initiative -- research university (2020--2026)'' of University of Warsaw. We are grateful to Erik Jan van Leeuwen for stimulating discussions on the topic.}

\begin{document}
\maketitle
\begin{abstract}
For a fixed graph $H$, \HFS\ asks whether at most $k$ edge
subdivisions can make an input graph contain no induced copy of $H$.
We establish a complexity dichotomy for every fixed nonempty $H$:
the problem is polynomial-time solvable if every component of $H$
is a subdivided star or a subdivided bistar and at most one component
is a subdivided bistar; it is NP-complete otherwise. In every
NP-complete case, the Exponential Time Hypothesis rules out a
$2^{o(k)}n^{O(1)}$-time algorithm. For every fixed $H$, the problem
admits a $2^{O(k)}n^{O(1)}$-time algorithm.
\end{abstract}

\ifthenelse{\boolean{shortver}}{%
\paragraph*{Full version.}
The full version of this paper, containing all proofs, is attached at the
end of this submission.
}{}

\newpage

\section{Introduction}

Graph modification problems ask whether an input graph can be transformed
into a prescribed graph class using a bounded number of operations. Their
complexity depends both on the permitted operation and on the target class.
A particularly natural choice is a \emph{hereditary} class, that is, a class
closed under taking induced subgraphs. Every such class can be described
by forbidden induced subgraphs, making it meaningful to seek a complexity
classification for each fixed forbidden graph $H$. This approach has led
to dichotomies for several graph modification operations, beginning with
the classical work on vertex and edge deletion
\cite{lewis1980node,DBLP:journals/siamcomp/Yannakakis81}.

We investigate \emph{edge subdivision} as a graph modification operation.
Subdividing an edge $uv$ replaces it by a path $uwv$, where $w$ is a new
vertex. In particular, the degrees of the original vertices are unchanged,
while distances and induced-subgraph structure may change. A graph is
\emph{$H$-free} if it contains no induced subgraph isomorphic to $H$.
For a fixed graph $H$, the problem we study is the following.

\medskip
\noindent
\HFS: Given a graph $G$ and a nonnegative integer $k$, decide whether
at most $k$ edge subdivisions can transform $G$ into an $H$-free graph.
\medskip

The choice of an induced-subgraph restriction as the target property is
important. Indeed, if a graph class is closed under taking minors, then
subdivision cannot help an input graph enter the class: the input graph
is a minor of every graph obtained from it by subdivision. The same
observation applies to classes closed under taking induced minors. By
contrast, the class of $H$-free graphs is not, in general, closed under
these operations. For example, subdividing an edge of a triangle
destroys that triangle, whereas subdividing the sole edge of $K_2$
creates an induced path $P_3$. Thus subdivisions can both eliminate and
create forbidden induced subgraphs, and the complexity is not determined
by a simple monotonicity argument.

\paragraph*{Related work.}
For a fixed forbidden induced subgraph $H$, Lewis and
Yannakakis~\cite{lewis1980node} established the classical hardness
picture for vertex deletion. Aravind, Sandeep, and
Sivadasan~\cite{DBLP:journals/siamdm/AravindSS17} proved dichotomies
for $H$-free edge deletion, completion, and editing, including
ETH-based parameterized lower bounds; Bliznets, Cygan, Komosa, and
Pilipczuk~\cite{BliznetsCKP18} subsequently developed approximation
hardness results and gap constructions for these operations.
The corresponding dichotomy for edge contraction was established by
Chakraborty and Sandeep~\cite{chakraborty2023contracting}.
Firbas and Sorge~\cite{firbas2024complexity} studied vertex splitting,
which replaces a vertex by two nonadjacent vertices and redistributes
its neighbours. Other work considers \emph{subgraph complementation}
(toggling adjacencies inside a chosen vertex set) and \emph{Seidel
switching} (toggling adjacencies between that set and its complement)
\cite{antony2022subgraph,antony2024cutting,antony2024switching}.
Table~\ref{tab:complexity-summary} recalls selected results for these
operations. It illustrates that the tractable forbidden graphs can
depend strongly on the modification operation; the classifications for
other operations do not determine the complexity of edge subdivision.

Edge subdivision is also fundamental in the study of topological minors
and has been used in hardness frameworks for restricted graph
classes~\cite{robertson1986graph,johnson2025complexity}. In those settings,
however, subdivision is not itself the budgeted operation used to reach a
hereditary target class.

\begin{table}[htbp]
\centering
\small
\begin{tabularx}{\textwidth}{@{}p{2.65cm}>{\raggedright\arraybackslash}X>{\raggedright\arraybackslash}X@{}}
\toprule
\textbf{Operation} & \textbf{Polynomial time} & \textbf{NP-complete} \\
\midrule
Vertex deletion
& $H=K_1$
& all other nonempty $H$~\cite{lewis1980node} \\
\addlinespace[2pt]
Edge deletion
& $|E(H)|\leq 1$
& $|E(H)|\geq 2$~\cite{DBLP:journals/siamdm/AravindSS17} \\
\addlinespace[2pt]
Edge completion
& $H$ has at most one nonedge
& $H$ has at least two nonedges~\cite{DBLP:journals/siamdm/AravindSS17} \\
\addlinespace[2pt]
Edge editing
& $|V(H)|\leq 2$
& $|V(H)|\geq 3$~\cite{DBLP:journals/siamdm/AravindSS17} \\
\addlinespace[2pt]
Edge contraction
& $H\in\{K_1,K_2\}$
& all other nonempty $H$~\cite{chakraborty2023contracting} \\
\addlinespace[2pt]
Subgraph complementation
& $P_t$ for $t\leq 4$, complete graphs, paw
& $K_{1,t}$ for $t\geq 5$; $P_t,C_t$ for $t\geq 7$, etc.~\cite{antony2022subgraph,antony2024cutting} \\
\addlinespace[2pt]
Seidel switching
& $|V(H)|\leq 2$, paw
& $P_t$ for $t\geq 10$; $C_t$ for $t\geq 7$, etc.~\cite{antony2024switching} \\
\addlinespace[2pt]
Vertex splitting
& $|V(H)|\leq 2$, $K_1+K_2$
& $2$-connected $H$~\cite{firbas2024complexity};
  $P_t$ for $t\geq 4$~\cite{jana2026hardnessvertexsplittingcographs} \\
\bottomrule
\end{tabularx}
\caption{Selected known complexity results for modification to $H$-free
graphs. The entries for subgraph complementation, switching, and vertex
splitting list representative cases rather than complete dichotomies.
For subgraph complementation and Seidel switching, one chooses a vertex
set rather than a budgeted sequence of operations.
Here $K_t$, $P_t$, and $C_t$ denote the complete graph, path, and cycle on
$t$ vertices, respectively; $K_{1,t}$ is the star with $t$ leaves.
The symbol $+$ denotes disjoint union, a \emph{paw} is a triangle with a
pendant edge, and a \emph{nonedge} is a pair of distinct nonadjacent
vertices.}
\label{tab:complexity-summary}
\end{table}
\FloatBarrier

\paragraph*{Motivation and main result.}
The effect of subdivision suggests that the structure of the forbidden
graph is decisive. Subdividing an edge introduces a vertex of degree
two without changing the degrees of the original vertices, but it can
disrupt a cycle or change the distance between vertices of degree at
least three (\emph{branching vertices}). For forests, the number and
arrangement of branching vertices therefore provide natural parameters
of the forbidden configuration. For example, an induced subdivided star
persists when one of its edges is subdivided, whereas the edge joining
the two branching vertices of a subdivided bistar plays a different
role. These observations lead to the following question: \emph{For which
fixed graphs $H$ can all induced copies of $H$ be eliminated by a
polynomial-time subdivision algorithm, and for which graphs is the
problem NP-complete?}

To state our answer, a \emph{star} is a tree with at most one vertex of
degree at least two; in particular, $K_1$ and $K_2$ are stars. A
\emph{bistar} is a tree with exactly two adjacent branching vertices,
all other vertices being leaves. A \emph{subdivided star} is obtained
by subdividing edges of a star zero or more times. A \emph{subdivided
bistar} is obtained from a bistar by subdividing only its edges other
than the edge joining its branching vertices; this central edge remains
unsubdivided. Our main result gives a complete classification,
including when $H$ is disconnected.

\Needspace{8\baselineskip}
\begin{restatable}{theorem}{DichotomyTheorem}
\label{thm:dichotomy}
Let $H$ be a fixed nonempty graph. If every connected component of $H$
is a subdivided star or a subdivided bistar, and at most one component
is a subdivided bistar, then \HFS\ is polynomial-time solvable.
In every other case, \HFS\ is NP-complete. Furthermore, assuming
the Exponential Time Hypothesis (ETH), every NP-complete case admits
no $2^{o(k)}n^{O(1)}$-time algorithm.
\end{restatable}

The polynomial-time cases rely on a persistence property: an induced
copy of a subdivided star or bistar cannot be destroyed by subdivision,
except potentially by subdividing the central edge of its bistar
component. Since the forbidden graph has at most one such component,
this yields a forced-choice algorithm. For the hardness cases, we
separately treat graphs containing triangles, triangle-free graphs
containing a cycle, and forests. The reductions use hardness results
for fixed-pattern edge deletion and \pname{Vertex Cover}, together with
constructions that control new induced copies created by subdivision.
For forests, the location and number of branching vertices govern the
hardness constructions, while additional arguments extend the results
from connected patterns to disconnected ones.

For every fixed $H$, there is also a straightforward single-exponential
parameterized algorithm. If an induced copy of $H$ is found, every
solution must subdivide at least one of its edges. Branching on these
at most $|E(H)|$ edges, with depth at most $k$, takes
$2^{O(k)}n^{O(1)}$ time; if $H$ has no edges, subdividing cannot destroy
an induced copy, so the instance can be decided directly. Hence the
ETH lower bounds match the single-exponential dependence on $k$.

The rest of the paper develops the preliminary tools and the
polynomial-time cases, proves hardness for graphs containing triangles
and for triangle-free graphs containing cycles, and then settles trees
and arbitrary forests.

\ifthenelse{\boolean{shortver}}{%
\section{Preliminaries}
\label{sec:preliminaries}
All graphs are finite, simple, and undirected, and the nonempty forbidden
graph $H$ is fixed. Budgets are nonnegative integers encoded in binary.
We write $G[X]$ for the subgraph induced by $X$, $G-X$ for vertex deletion,
and $G-F$ for deletion of an edge set $F$. The complement, disjoint union,
and join are denoted by $\overline G$, $G\mathbin{\dot\cup}G'$, and
$G\vee G'$, respectively. Two disjoint vertex sets are \emph{complete}
if all edges between them are present, and \emph{anticomplete} if none is.
A \emph{branching vertex} has degree at least three. The \emph{core}
$\core(G)$ is the $2$-core, obtained by repeatedly deleting vertices of
degree at most one. A \emph{linear forest} is a disjoint union of paths
and isolated vertices. The girth is the length of a shortest cycle, or
$\infty$ for a forest. A graph is $t$-connected if it has more than $t$
vertices and remains connected after deletion of fewer than $t$ vertices.
We use $[t]=\{1,\ldots,t\}$ and standard notation for degrees, neighbourhoods,
distances, and complete bipartite graphs $K_{a,b}$.

Subdivisions may act on edges created by earlier operations. An original
edge is \emph{touched} if it or a descendant is subdivided. A sequence of
at most $k$ operations touches at most $k$ original edges, and the final
graph induced by the original vertices is obtained by deleting precisely
those edges. Unless stated otherwise, $n,m$ refer to the input graph.

The problems considered here belong to NP. Put $h:=|V(H)|$. Any solution
can be truncated to at most $h$ subdivisions per original edge: an induced
$H$ in the truncated graph could also be embedded in the longer subdivision,
since at most $h$ selected vertices leave a gap on every replacement path
with $h$ internal vertices, and lengthening at such a gap preserves their
induced adjacencies. Thus polynomial-size certificates suffice even for
binary budgets. Below we establish only NP-hardness.

We use the following source problems and hardness results. The problem
\HFD\ asks whether deleting at most $k$ edges makes a graph $H$-free.
A \emph{vertex cover} meets every edge.

\begin{proposition}[\cite{DBLP:journals/siamdm/AravindSS17}]
\label{pro:deletion}
If $H$ has at least two edges, then \HFD\ is NP-complete and, assuming
ETH, admits no $2^{o(k)}n^{O(1)}$-time algorithm.
\end{proposition}

\begin{proposition}[\cite{komusiewicz2018tight,poljak1974note}]
\label{pro:vc}
For every fixed $\ell\geq3$, \pname{Vertex Cover} is NP-complete on graphs
of girth at least $\ell$ and, assuming ETH, admits no $2^{o(n+m)}$-time algorithm.
\end{proposition}

The following consequence of the sandwich-completion reductions of
Bliznets et al.~\cite[Lemmas~3.5 and~3.6]{BliznetsCKP18} is proved by
complementation in the attached full version.
\begin{proposition}
\label{pro:deletion-gap}
Let $Q$ be a fixed graph with at least two edges, let $\overline Q$ be $3$-connected, and
let $c\geq1$ be a fixed integer. A polynomial-time reduction maps a
$3$-SAT instance with $N$ variables and clauses in total to $(B,p)$,
where $p=O(N)$, such that satisfiable instances admit a $Q$-free edge
deletion set of size at most $p$, whereas unsatisfiable instances admit
none of size at most $cp$.
\end{proposition}

All reductions have polynomial output size. We may assume that a source
\pname{Vertex Cover} budget is at most the number of source vertices.
Consequently, preserving that budget, or multiplying it by a fixed
constant, transfers the lower bound in \cref{pro:vc} to
$2^{o(k)}n^{O(1)}$ for the target. For ETH background, see
Cygan et al.~\cite{cygan2015parameterized}.

}{%
\section{Preliminaries}
\label{sec:preliminaries}
All graphs considered in this paper are finite, simple, and undirected.
The forbidden graph $H$ is fixed and nonempty; all budgets are
nonnegative integers encoded in binary. If the empty graph is allowed as a forbidden pattern, every
instance is negative, so we exclude that degenerate convention.
A graph is $H$-\emph{free} if it has no induced subgraph isomorphic to $H$.
A \emph{clique} is a set of pairwise adjacent vertices, and an
\emph{independent set} is a set of pairwise nonadjacent vertices.
Two vertex sets are \emph{complete} to one another if every pair of
vertices across the sets is adjacent, and \emph{anticomplete} if none is.
We write $G\mathbin{\dot\cup}G'$ for disjoint union, $G\vee G'$
for the join of two disjoint graphs (adding all edges between them),
and $\overline G$ for the complement of $G$. A \emph{linear forest} is a disjoint union
of paths and isolated vertices. A graph is $t$-\emph{connected} if it
has more than $t$ vertices and remains connected after deletion of
fewer than $t$ vertices.

For a graph $G$ and $X\subseteq V(G)$, we denote by $G[X]$ the subgraph induced by $X$
and by $G-X$ the graph $G[V(G)\setminus X]$; we abbreviate
$G-\{v\}$ as $G-v$.
For $F\subseteq E(G)$, the graph $G-F$ is obtained by deleting all edges in $F$.
By $\overline{G}$ we denote the complement of $G$.
We denote the neighbours of $v$ in $G$ by $N_G(v)$ and its degree by
$\deg_G(v):=|N_G(v)|$.
For vertices $u,v\in V(G)$, we denote by $\dist_G(u,v)$ the
number of edges in a shortest $u$--$v$ path.
We denote by $\delta(G)$ the minimum degree of $G$.
The \emph{girth} of a graph $G$, denoted by $\girth(G)$, is the length of a shortest cycle in $G$; we put $\girth(G)=\infty$ for forests.
For two disjoint sets $X_1,X_2\subseteq V(G)$, we denote by $E_G(X_1,X_2)$ the set of edges of $G$ with one endpoint in $X_1$ and the other in $X_2$.
If the graph $G$ is clear from the context, we omit the subscript and write simply $\deg(v)$, $\dist(u,v)$, and $E(X_1,X_2)$.

A vertex of degree at least three is called a \emph{branching vertex}.
The \emph{core} of a graph $G$, denoted by $\core(G)$, is its $2$-core,
obtained by repeatedly deleting vertices of degree at most one.  For a positive
integer $t$, we denote by $C_t$, $K_t$, and $P_t$ the cycle, complete graph,
and path on $t$ vertices, respectively (with cycles defined for $t\geq3$).
For positive $a,b$, the graph $K_{a,b}$ is the complete bipartite
graph with parts of sizes $a$ and $b$. A \emph{vertex cover} is a
vertex set meeting every edge.
We abbreviate $\{1,\ldots,t\}$ by $[t]$.
Unless specified otherwise, $n$ and $m$ denote the numbers of vertices and
edges of the input graph, not of a graph after subdivision. Subdivisions may
be repeated: an operation may act on an edge created by an earlier operation.
An original edge is \emph{touched} if it or any of its descendants is subdivided;
its subdivision count is the number of operations on that original edge and
its descendants.

For each fixed forbidden graph, the decision problems considered here
belong to NP. Indeed, when $h:=|V(H)|$, a solution can be truncated
to at most $h$ subdivisions per original edge without introducing a
new induced copy of $H$. Indeed, any set of at most $h$ vertices
of a path with $h$ internal vertices leaves a gap; lengthening the
path at that gap preserves the induced adjacencies of the selected
vertices, including any original endpoints. Thus a yes-instance has a polynomial-size certificate even
when $k$ is encoded in binary. We establish only NP-hardness in the
NP-completeness proofs below.

Our reductions use \pname{$H'$-free Edge Deletion}, \pname{Vertex Cover}
(sometimes with girth restrictions), and a gap reduction originating
from \pname{3-SAT}. For a fixed graph $H'$, \pname{$H'$-free Edge Deletion} is defined as follows: given a graph $G$ and an integer $k$, decide whether 
deleting at most $k$ edges from $G$ results in an $H'$-free graph.

\begin{proposition}[\cite{DBLP:journals/siamdm/AravindSS17}]
\label{pro:deletion}
\HFD\ is \NPC\ if $H$ has at least two edges.
Furthermore, assuming the ETH, the problem cannot be solved in time
$2^{o(k)}\cdot n^{O(1)}$.
\end{proposition}

For a graph $G$ and a set $F\subseteq E(\overline G)$, let $G+F$
denote the graph obtained by adding all pairs in $F$ as edges.
For a fixed graph $J$, an instance of \pname{Sandwich $J$-free Edge Completion} consists
of a graph $G$ and a set $A\subseteq E(\overline G)$ of
\emph{fillable nonedges}. The question is whether there exists
$A'\subseteq A$ such that $G+A'$ is $J$-free; no bound on
$|A'|$ is imposed.

\begin{proposition}[{\cite[Lemma~3.5]{BliznetsCKP18}}]
\label{pro:sandwich-completion}
Let $J$ be a fixed $3$-connected graph with at least two nonedges.
There is a polynomial-time reduction that, given a $3$-SAT formula
$\varphi$ (a Boolean formula with clauses of size at most three)
with $N$ variables and clauses in total, constructs a graph
$G_0$ and a set $A\subseteq E(\overline{G_0})$ such that
\[
|V(G_0)|+|E(G_0)|+|A|=O(N)
\]
and
\[
\varphi \text{ is satisfiable}
\quad\Longleftrightarrow\quad
\text{there exists } A'\subseteq A
\text{ such that } G_0+A' \text{ is } J\text{-free}.
\]
\end{proposition}

\begin{proposition}[{\cite[Lemma~3.6]{BliznetsCKP18}}]
\label{pro:completion-gap}
Let $J$ be a fixed $3$-connected graph, and let $q$ be a fixed
polynomial taking integer values on positive integers, with
$q(t)\geq t$ for every positive integer $t$.
There is a polynomial-time reduction that, given an instance
$(G_0,A)$ of \pname{Sandwich $J$-free Edge Completion} with
$k:=|A|\geq1$, constructs a graph $G$ such that:
\begin{enumerate}
\item if there exists $A'\subseteq A$ such that $G_0+A'$ is
$J$-free, then there exists $F\subseteq E(\overline G)$ with
$|F|\leq k$ such that $G+F$ is $J$-free;
\item if no such $A'$ exists, then $G+F$ contains an induced
copy of $J$ for every $F\subseteq E(\overline G)$ with
$|F|\leq q(k)$.
\end{enumerate}
\end{proposition}

The following deletion gap follows from
\cref{pro:sandwich-completion,pro:completion-gap} by complementation.

\begin{proposition}
\label{pro:deletion-gap}
Let $Q$ be a fixed graph with at least two edges such that $\overline Q$
is $3$-connected, and let $c\geq1$ be a fixed integer. There is a
polynomial-time reduction that, given a $3$-SAT instance with $N$ variables
and clauses in total, produces a graph $B$ and a nonnegative integer
$p=O(N)$ such that:
\begin{enumerate}
\item if the instance is satisfiable, there exists $F\subseteq E(B)$
with $|F|\leq p$ such that $B-F$ is $Q$-free;
\item if the instance is unsatisfiable, $B-F$ contains an induced copy of
$Q$ for every $F\subseteq E(B)$ with $|F|\leq cp$.
\end{enumerate}
\end{proposition}

\begin{proof}
Let $\varphi$ be the input formula, and put $J:=\overline Q$.
Since $J$ is $3$-connected and has at least two nonedges,
\cref{pro:sandwich-completion} yields in polynomial time an
instance $(G_0,A)$ of \pname{Sandwich $J$-free Edge Completion}
that is a yes-instance if and only if $\varphi$ is satisfiable.
Set $p:=|A|$; the proposition guarantees that $p=O(N)$.

If $p\geq1$, apply \cref{pro:completion-gap} to $(G_0,A)$
with $q(t):=ct$. This yields a graph $G$ admitting a $J$-free
completion with at most $p$ additions when $\varphi$ is satisfiable,
but no such completion with at most $cp$ additions when $\varphi$
is unsatisfiable.
If $p=0$, instead set $G:=G_0$. Since $A=\emptyset$,
\cref{pro:sandwich-completion} implies that $G$ is $J$-free
if and only if $\varphi$ is satisfiable. As $p=cp=0$, the same
completion gap holds in this case.

Output $(B,p)$, where $B:=\overline G$.
For every $F\subseteq E(B)=E(\overline G)$, we have
$\overline{B-F}=G+F$. Consequently,
\[
B-F \text{ is } Q\text{-free}
\quad\Longleftrightarrow\quad
G+F \text{ is } J\text{-free}.
\]
If $\varphi$ is satisfiable, the completion gap therefore gives
a set $F\subseteq E(B)$ with $|F|\leq p$ such that $B-F$ is
$Q$-free. If $\varphi$ is unsatisfiable, then for every
$F\subseteq E(B)$ with $|F|\leq cp$, the graph $G+F$ contains
an induced copy of $J$, whose vertex set induces a copy of $Q$
in $B-F$.

Both reductions run in polynomial time for fixed $Q$ and $c$,
and complementation is polynomial-time. Thus the entire
construction is polynomial-time and retains the bound $p=O(N)$.
\end{proof}

\begin{proposition}[\cite{komusiewicz2018tight,poljak1974note}]
\label{pro:vc}
For every fixed integer $\ell\geq 3$, \pname{Vertex Cover} is \NPC\ on graphs of girth at least $\ell$.
Moreover, assuming the ETH, \pname{Vertex Cover} on graphs of girth at least $\ell$ cannot be solved in time
$2^{o(n+m)}$.
\end{proposition}

All reductions below are polynomial-time many-one reductions. In a
\pname{Vertex Cover} source instance we may assume that the budget is at most
the number of source vertices, since larger budgets are trivially positive.
Our output graphs may be dense; only polynomial output size is needed. A
parameter-preserving (or constant-factor parameter) reduction with this
budget bound transfers the $2^{o(n+m)}$ lower bound in \cref{pro:vc} to a
$2^{o(k)}n^{O(1)}$ lower bound for the target problem.

For background on the Exponential Time Hypothesis and on techniques for proving parameterized lower bounds under the ETH, we refer to the textbook of Cygan et al.~\cite{cygan2015parameterized}.

}

\ifthenelse{\boolean{shortver}}{%
\section{Polynomial cases and initial hardness results}
\label{sec:initial-results}

\begin{observation}
\label{obs:poly}
If every component of $H$ is a subdivided star or a subdivided bistar,
with at most one subdivided-bistar component, then \HFS\ is
polynomial-time solvable.
\end{observation}
\begin{proof}
An induced copy survives subdivision of an edge outside it. It also
survives subdivision on a pendant path, by using the new vertex and
truncating the path at its original length. Thus only subdivision of
the central edge of its bistar component can destroy it.
If no bistar component exists, test whether the input is already $H$-free.
Otherwise, repeatedly find an induced $H$ and subdivide this forced
central edge, decreasing the budget by one. Operations on different
original edges commute, so a forced subdivision can be performed first.
Accept when the graph is $H$-free, and reject if the budget is exhausted
while a copy remains.
Every chosen edge joins branching vertices and hence is original;
once subdivided, it cannot be chosen again. There are at most $m$
iterations and $n+m$ vertices, giving a polynomial-time algorithm.
\end{proof}

\begin{theorem}
\label{thm:easy}
Let $H$ satisfy either of the following conditions:
\begin{enumerate}
\item $\delta(H)\geq2$, and the neighbours of every degree-$2$ vertex
are adjacent;
\item the branching vertices induce at least two edges.
\end{enumerate}
Then \HFS\ is NP-complete and, assuming ETH, admits no
$2^{o(k)}n^{O(1)}$-time algorithm.
\end{theorem}

For the first condition, subdivision and edge deletion are equivalent
on the same instance. A successful subdivision sequence gives a deletion
set by taking its touched original edges. Conversely, subdivide each
edge of a successful deletion set once. Every new vertex has degree two
and nonadjacent neighbours, so it cannot occur in an induced $H$.
\Cref{pro:deletion} therefore applies. We prove the second condition
using the following general construction.

\paragraph*{Private completion.}
Fix $X\subseteq V(H)$ and put $Y:=V(H)\setminus X$ and $H_X:=H[X]$.
For a base graph $B$ and budget $k$, construct $\Gamma_X(B,k)$ as follows.
For every injection $\varphi:X\to V(B)$, add $k+1$ private copies of
$H[Y]$. In each copy, join the representative of $y\in Y$ to
$\varphi(x)$ precisely when $xy\in E(H)$. Distinct private copies are
anticomplete, and no other edges are added. Their internal and attachment
edge sets are pairwise disjoint, and each private representative of $y$
has degree $\deg_H(y)$. We use all injections because deleting base edges
can create induced copies of $H_X$.
For $F\subseteq E(B)$, let $\Gamma_X^F(B,k)$ result from subdividing each
edge of $F$ once; its original base vertices induce $B-F$.
See \cref{fig:private-completion} for an example.

\begin{figure}[t]
    \centering
    \begin{tikzpicture}[every node/.style={draw,circle,fill=white,inner sep=0pt,minimum size=8pt},every loop/.style={}]
\node[fill=blue,opacity=0.5] (v1) at (0,0) {};
\node[fill=gray,opacity=0.5] (v2) at (1,0) {};
\node[fill=gray,opacity=0.5] (v3) at (1,1) {};
\node[fill=orange,opacity=0.5] (v4) at (0,1) {};
\node[fill=green,opacity=0.5] (v5) at (0.5,1.5) {};

\draw (v1)--(v2)--(v3)--(v4)--(v1);
\draw (v3)--(v5)--(v4);

\node[draw=none,fill=none] (h) at (-1.5,0.5) {$H=$};

\node (u1) at (-0.5,-1.5) {};
\node (u2) at (0.5,-1.5) {};
\node (u3) at (1.5,-1.5) {};
\draw (u1)--(u2)--(u3);
\node[draw=none,fill=none] (b) at (-1.5,-1.5) {$B=$};


\node (w1) at (6,0) {};
\node (w2) at (7,0) {};
\node (w3) at (8,0) {};
\draw (w1)--(w2)--(w3);

\node[fill=blue,opacity=0.5] (v11) at (3.25,1.5) {};
\node[fill=orange,opacity=0.5] (v12) at (4,1.5) {};
\node[fill=green,opacity=0.5] (v13) at (4.75,1.5) {};
\draw (v11)--(v12)--(v13);
\draw (v11)--(w1);
\draw (v12)--(w2);
\draw (v13)--(w2);

\node[fill=blue,opacity=0.5] (w11) at (3.25,0.9) {};
\node[fill=orange,opacity=0.5] (w12) at (4,0.9) {};
\node[fill=green,opacity=0.5] (w13) at (4.75,0.9) {};
\draw (w11)--(w12)--(w13);
\draw (w11)--(w1);
\draw (w12)--(w2);
\draw (w13)--(w2);


\node[fill=blue,opacity=0.5] (v21) at (6.25,2) {};
\node[fill=orange,opacity=0.5] (v22) at (7,2) {};
\node[fill=green,opacity=0.5] (v23) at (7.75,2) {};
 \draw (v21)--(v22)--(v23);
\draw (v21) to [bend right = 15] (w1);
 \draw (v22)--(w3);
\draw (v23) to [bend left = 15] (w3);

\node[fill=blue,opacity=0.5] (w21) at (6.25,1.5) {};
\node[fill=orange,opacity=0.5] (w22) at (7,1.5) {};
\node[fill=green,opacity=0.5] (w23) at (7.75,1.5) {};
 \draw (w21)--(w22)--(w23);
 \draw (w21)--(w1);
 \draw (w22)--(w3);
 \draw (w23)--(w3);

\node[fill=blue,opacity=0.5] (v31) at (9.25,1.5) {};
\node[fill=orange,opacity=0.5] (v32) at (10,1.5) {};
\node[fill=green,opacity=0.5] (v33) at (10.75,1.5) {};
\draw (v31)--(v32)--(v33);
 \draw (v31)--(w2);
 \draw (v32) to [bend left=15] (w3);
\draw (v33)  to [bend left=30] (w3);

\node[fill=blue,opacity=0.5] (w31) at (9.25,0.9) {};
\node[fill=orange,opacity=0.5] (w32) at (10,0.9) {};
\node[fill=green,opacity=0.5] (w33) at (10.75,0.9) {};
 \draw (w31)--(w32)--(w33);
 \draw (w31)--(w2);
 \draw (w32)--(w3);
 \draw (w33)--(w3);

 \node[fill=green,opacity=0.5] (v41) at (3.25,-1.5) {};
\node[fill=orange,opacity=0.5] (v42) at (4,-1.5) {};
\node[fill=blue,opacity=0.5] (v43) at (4.75,-1.5) {};
\draw (v41)--(v42)--(v43);
 \draw (v41)--(w1);
 \draw (v42) to [bend right=20] (w1);
 \draw (v43)--(w2);

\node[fill=green,opacity=0.5] (w41) at (3.25,-0.9) {};
\node[fill=orange,opacity=0.5] (w42) at (4,-0.9) {};
\node[fill=blue,opacity=0.5] (w43) at (4.75,-0.9) {};
\draw (w41)--(w42)--(w43);
\draw (w41)--(w1);
\draw (w42)--(w1);
\draw (w43)--(w2);

\node[fill=green,opacity=0.5] (v51) at (6.25,-2) {};
\node[fill=orange,opacity=0.5] (v52) at (7,-2) {};
\node[fill=blue,opacity=0.5] (v53) at (7.75,-2) {};
 \draw (v51)--(v52)--(v53);
 \draw (v51) to [bend left = 15] (w1);
 \draw (v52)--(w1);
 \draw (v53) to [bend right = 15] (w3);

\node[fill=green,opacity=0.5] (w51) at (6.25,-1.5) {};
\node[fill=orange,opacity=0.5] (w52) at (7,-1.5) {};
\node[fill=blue,opacity=0.5] (w53) at (7.75,-1.5) {};
 \draw (w51)--(w52)--(w53);
 \draw (w51)--(w1);
  \draw (w52)--(w1);
  \draw (w53)--(w3);

  \node[fill=green,opacity=0.5] (v61) at (9.25,-1.5) {};
\node[fill=orange,opacity=0.5] (v62) at (10,-1.5) {};
\node[fill=blue,opacity=0.5] (v63) at (10.75,-1.5) {};
\draw (v61)--(v62)--(v63);
 \draw (v61)--(w2);
 \draw (v62) -- (w2);
 \draw (v63) -- (w3);

\node[fill=green,opacity=0.5] (w61) at (9.25,-0.9) {};
\node[fill=orange,opacity=0.5] (w62) at (10,-0.9) {};
\node[fill=blue,opacity=0.5] (w63) at (10.75,-0.9) {};
 \draw (w61)--(w62)--(w63);
  \draw (w61)--(w2);
  \draw (w62)--(w2);
  \draw (w63)--(w3);
\end{tikzpicture}
    \caption{Example of the private completion construction (right) for $k=1$ and graphs $H$ and $B$ shown on the left. Gray vertices of $H$ denote the set $X$. Other colors denote the corresponding vertices of $Y$ in the copies of $H[Y]$.}
    \label{fig:private-completion}
\end{figure}

\begin{lemma}[Private-completion lemma]
\label{lem:private-completion}
Suppose $H_X$ has at least two edges and the following property holds:
for every $B,k,F$ and every induced copy $J$ of $H$ in
$\Gamma_X^F(B,k)$, there is an isomorphism $\psi:H\to J$ with
$\psi(X)\subseteq V(B)$. Then \HFS\ is NP-complete and, assuming ETH,
admits no $2^{o(k)}n^{O(1)}$-time algorithm.
\end{lemma}
\begin{proof}
Reduce from \pname{$H_X$-free Edge Deletion}, mapping $(B,k)$ to
$(\Gamma_X(B,k),k)$. Instances with $k\geq |E(B)|$ are trivially positive;
otherwise the construction is polynomial because $H$ is fixed.
If $B-F$ is $H_X$-free with $|F|\leq k$, subdivide $F$ once. The stated
property excludes an induced $H$ in the result.
Conversely, let at most $k$ subdivisions make the target $H$-free, and
let $F$ be the touched base edges. If $B-F$ contains $H_X$, choose its
embedding $\varphi$. Among the $k+1$ private copies for $\varphi$, one
has all internal and attachment edges untouched. Together with
$\varphi(X)$ it induces $H$, a contradiction. Thus $B-F$ is $H_X$-free.
The budget is preserved, so \cref{pro:deletion} gives both conclusions.
\end{proof}

\begin{lemma}
\label{lem:degree-3}
If the branching vertices of $H$ induce at least two edges, then
\HFS\ is NP-complete and, assuming ETH, admits no
$2^{o(k)}n^{O(1)}$-time algorithm.
\end{lemma}
\begin{proof}
Apply \cref{lem:private-completion} with $X$ equal to the branching
vertices. Every vertex outside the base has degree at most two, so
every representative of a vertex in $X$ lies in the base.
\end{proof}
This completes the proof of \cref{thm:easy}.

\begin{corollary}
\label{cor:two-bistar-components}
If $H$ has at least two subdivided-bistar components, then \HFS\ is
NP-complete and, assuming ETH, admits no $2^{o(k)}n^{O(1)}$-time algorithm.
\end{corollary}
\begin{proof}
The two central edges join branching vertices, so \cref{lem:degree-3} applies.
\end{proof}

}{%
\section{Polynomial cases and initial hardness results}
\label{sec:initial-results}

We first establish the polynomial-time side of \cref{thm:dichotomy}.

\begin{observation}
\label{obs:poly}
If every component of $H$ is a subdivided star or a subdivided bistar,
with at most one subdivided-bistar component, then \HFS\ is
polynomial-time solvable.
\end{observation}

\begin{proof}
An induced copy of $H$ survives subdivision of any edge outside it.
It also survives subdivision of any edge on one of its pendant paths:
a new subdivision vertex can be used while truncating the path at
its original length. Consequently, the only operation that can
destroy a given induced copy is subdivision of the central edge of
its bistar component, if such a component exists.

If there is no bistar component, we only need to test whether the
input is already $H$-free. Otherwise, repeatedly find an induced
copy of $H$ and subdivide the central edge of its bistar component,
decreasing the budget by one. This step is forced in every solution;
operations on different original edges commute, so the chosen
subdivision can be performed first. Accept once the graph is
$H$-free, and reject if the budget is exhausted while a copy remains.

Each chosen central edge has endpoints of degree at least three.
Subdivision vertices have degree two, so every chosen edge is an
original edge and no edge is chosen twice. There are at most $m$
iterations, and at most $n+m$ vertices at any stage. Since $H$ is
fixed, finding an induced copy is polynomial-time.
\end{proof}

We next prove the following general hardness criterion.

\begin{theorem}
\label{thm:easy}
Let $H$ be a graph satisfying at least one of the following properties:
\begin{itemize}
  \item $H$ has minimum degree at least $2$, and the neighbours of every
        degree-$2$ vertex induce a $K_2$;
  \item the vertices of degree at least $3$ induce a graph with at least
        two edges.
\end{itemize}
Then \HFS\ is NP-complete. Furthermore, assuming ETH, it admits no
$2^{o(k)}n^{O(1)}$-time algorithm.
\end{theorem}

First we handle the first case of \cref{thm:easy}, by showing polynomial
equivalence with the corresponding edge-deletion problem.

\begin{lemma}
\label{lem:edge-deletion-eq}
Let $H$ be a graph with minimum degree at least $2$ such that the neighbourhood
of every degree-$2$ vertex induces a $K_2$.
Then \HFS\ is polynomially equivalent to \HFD.
\end{lemma}

\begin{proof}
Suppose first that at most $k$ subdivisions turn $G$ into an $H$-free graph
$G'$, and let $F$ be the set of original edges touched by the sequence.  Then
$|F|\leq k$ and $G'[V(G)]=G-F$, so $G-F$ is $H$-free.

Conversely, let $F\subseteq E(G)$, $|F|\leq k$, and suppose that $G-F$ is
$H$-free.  Subdivide every edge of $F$ once.  Each new vertex has degree $2$
and has two nonadjacent neighbours.  It can therefore represent neither a
vertex of degree at least $3$ in $H$ nor a degree-$2$ vertex of $H$; vertices
of smaller degree do not exist because $\delta(H)\geq2$.  Hence every induced
copy of $H$ would lie in the original vertex set, which induces $G-F$.
\end{proof}

We next handle the second case. We first give a precise definition of the
private-completion construction used below.

\paragraph*{The private-completion construction.}
Intuitively, for an integer $k$, a graph $H$, a partition $(X,Y)$ of $V(H)$, and for a base graph $B$, we want to create many ``copies'' of $H$ that intersect $B$.
For every possible injective mapping of $X$ to $V(B)$ (the mapping does not have to preserve the edges of $H[X]$, hence, we do not speak about proper copies here), we introduce $k+1$ disjoint copies of $H[Y]$ and connect the vertices to vertices in $B$ according to the edges of $H$.
See \cref{fig:private-completion} for an example.
Now let us describe this construction formally.

\begin{figure}[t]
    \centering
    \begin{tikzpicture}[every node/.style={draw,circle,fill=white,inner sep=0pt,minimum size=8pt},every loop/.style={}]
\node[fill=blue,opacity=0.5] (v1) at (0,0) {};
\node[fill=gray,opacity=0.5] (v2) at (1,0) {};
\node[fill=gray,opacity=0.5] (v3) at (1,1) {};
\node[fill=orange,opacity=0.5] (v4) at (0,1) {};
\node[fill=green,opacity=0.5] (v5) at (0.5,1.5) {};

\draw (v1)--(v2)--(v3)--(v4)--(v1);
\draw (v3)--(v5)--(v4);

\node[draw=none,fill=none] (h) at (-1.5,0.5) {$H=$};

\node (u1) at (-0.5,-1.5) {};
\node (u2) at (0.5,-1.5) {};
\node (u3) at (1.5,-1.5) {};
\draw (u1)--(u2)--(u3);
\node[draw=none,fill=none] (b) at (-1.5,-1.5) {$B=$};


\node (w1) at (6,0) {};
\node (w2) at (7,0) {};
\node (w3) at (8,0) {};
\draw (w1)--(w2)--(w3);

\node[fill=blue,opacity=0.5] (v11) at (3.25,1.5) {};
\node[fill=orange,opacity=0.5] (v12) at (4,1.5) {};
\node[fill=green,opacity=0.5] (v13) at (4.75,1.5) {};
\draw (v11)--(v12)--(v13);
\draw (v11)--(w1);
\draw (v12)--(w2);
\draw (v13)--(w2);

\node[fill=blue,opacity=0.5] (w11) at (3.25,0.9) {};
\node[fill=orange,opacity=0.5] (w12) at (4,0.9) {};
\node[fill=green,opacity=0.5] (w13) at (4.75,0.9) {};
\draw (w11)--(w12)--(w13);
\draw (w11)--(w1);
\draw (w12)--(w2);
\draw (w13)--(w2);


\node[fill=blue,opacity=0.5] (v21) at (6.25,2) {};
\node[fill=orange,opacity=0.5] (v22) at (7,2) {};
\node[fill=green,opacity=0.5] (v23) at (7.75,2) {};
 \draw (v21)--(v22)--(v23);
\draw (v21) to [bend right = 15] (w1);
 \draw (v22)--(w3);
\draw (v23) to [bend left = 15] (w3);

\node[fill=blue,opacity=0.5] (w21) at (6.25,1.5) {};
\node[fill=orange,opacity=0.5] (w22) at (7,1.5) {};
\node[fill=green,opacity=0.5] (w23) at (7.75,1.5) {};
 \draw (w21)--(w22)--(w23);
 \draw (w21)--(w1);
 \draw (w22)--(w3);
 \draw (w23)--(w3);

\node[fill=blue,opacity=0.5] (v31) at (9.25,1.5) {};
\node[fill=orange,opacity=0.5] (v32) at (10,1.5) {};
\node[fill=green,opacity=0.5] (v33) at (10.75,1.5) {};
\draw (v31)--(v32)--(v33);
 \draw (v31)--(w2);
 \draw (v32) to [bend left=15] (w3);
\draw (v33)  to [bend left=30] (w3);

\node[fill=blue,opacity=0.5] (w31) at (9.25,0.9) {};
\node[fill=orange,opacity=0.5] (w32) at (10,0.9) {};
\node[fill=green,opacity=0.5] (w33) at (10.75,0.9) {};
 \draw (w31)--(w32)--(w33);
 \draw (w31)--(w2);
 \draw (w32)--(w3);
 \draw (w33)--(w3);

 \node[fill=green,opacity=0.5] (v41) at (3.25,-1.5) {};
\node[fill=orange,opacity=0.5] (v42) at (4,-1.5) {};
\node[fill=blue,opacity=0.5] (v43) at (4.75,-1.5) {};
\draw (v41)--(v42)--(v43);
 \draw (v41)--(w1);
 \draw (v42) to [bend right=20] (w1);
 \draw (v43)--(w2);

\node[fill=green,opacity=0.5] (w41) at (3.25,-0.9) {};
\node[fill=orange,opacity=0.5] (w42) at (4,-0.9) {};
\node[fill=blue,opacity=0.5] (w43) at (4.75,-0.9) {};
\draw (w41)--(w42)--(w43);
\draw (w41)--(w1);
\draw (w42)--(w1);
\draw (w43)--(w2);

\node[fill=green,opacity=0.5] (v51) at (6.25,-2) {};
\node[fill=orange,opacity=0.5] (v52) at (7,-2) {};
\node[fill=blue,opacity=0.5] (v53) at (7.75,-2) {};
 \draw (v51)--(v52)--(v53);
 \draw (v51) to [bend left = 15] (w1);
 \draw (v52)--(w1);
 \draw (v53) to [bend right = 15] (w3);

\node[fill=green,opacity=0.5] (w51) at (6.25,-1.5) {};
\node[fill=orange,opacity=0.5] (w52) at (7,-1.5) {};
\node[fill=blue,opacity=0.5] (w53) at (7.75,-1.5) {};
 \draw (w51)--(w52)--(w53);
 \draw (w51)--(w1);
  \draw (w52)--(w1);
  \draw (w53)--(w3);

  \node[fill=green,opacity=0.5] (v61) at (9.25,-1.5) {};
\node[fill=orange,opacity=0.5] (v62) at (10,-1.5) {};
\node[fill=blue,opacity=0.5] (v63) at (10.75,-1.5) {};
\draw (v61)--(v62)--(v63);
 \draw (v61)--(w2);
 \draw (v62) -- (w2);
 \draw (v63) -- (w3);

\node[fill=green,opacity=0.5] (w61) at (9.25,-0.9) {};
\node[fill=orange,opacity=0.5] (w62) at (10,-0.9) {};
\node[fill=blue,opacity=0.5] (w63) at (10.75,-0.9) {};
 \draw (w61)--(w62)--(w63);
  \draw (w61)--(w2);
  \draw (w62)--(w2);
  \draw (w63)--(w3);
\end{tikzpicture}
    \caption{Example of the private completion construction (right) for $k=1$ and graphs $H$ and $B$ shown on the left. Gray vertices of $H$ denote the set $X$. Other colors denote the corresponding vertices of $Y$ in the copies of $H[Y]$.}
    \label{fig:private-completion}
\end{figure}

Fix a graph $H$ and a set $X\subseteq V(H)$, and put
$Y:=V(H)\setminus X$ and $H_X:=H[X]$. Let $B$ be a graph, called the
\emph{base graph}, and let $k$ be a nonnegative integer. Write
\[
  \mathcal I_X(B):=\{\varphi:X\to V(B):\varphi\text{ is injective}\}.
\]
For every $\varphi\in\mathcal I_X(B)$ and every
$i\in\{1,\ldots,k+1\}$, introduce a private graph $Q_{\varphi,i}$ with
\begin{align*}
  V(Q_{\varphi,i})&:=\{y_{\varphi,i}:y\in Y\},\\
  E(Q_{\varphi,i})&:=\{y_{\varphi,i}z_{\varphi,i}:yz\in E(H[Y])\}.
\end{align*}
All these vertex sets are pairwise disjoint and disjoint from $V(B)$.
Define the set of \emph{private edges} belonging to this copy by
\[
 A_{\varphi,i}:=E(Q_{\varphi,i})\cup
 \{\varphi(x)y_{\varphi,i}:x\in X,\ y\in Y,\ xy\in E(H)\}.
\]
Thus $A_{\varphi,i}$ consists of the internal edges of $Q_{\varphi,i}$
and its attachment edges to the base. Define $\Gamma_X(B,k)$ by
\begin{align*}
 V(\Gamma_X(B,k))&:=V(B)\cup
   \bigcup_{\varphi\in\mathcal I_X(B)}\ \bigcup_{i=1}^{k+1}
       V(Q_{\varphi,i}),\\
 E(\Gamma_X(B,k))&:=E(B)\cup
   \bigcup_{\varphi\in\mathcal I_X(B)}\ \bigcup_{i=1}^{k+1}
       A_{\varphi,i}.
\end{align*}
There are no other edges. In particular, the construction does not change
any adjacency within $V(B)$, distinct private copies are anticomplete,
and a private vertex has no base neighbour outside its prescribed image
of $X$. The sets $A_{\varphi,i}$ are pairwise disjoint and are disjoint
from $E(B)$, even when different copies have common attachment vertices.
Moreover,
\[
 \deg_{\Gamma_X(B,k)}(y_{\varphi,i})=\deg_H(y)
 \qquad(y\in Y).
\]
The construction uses \emph{all} injections $\varphi$, not just those
that induce a copy of $H_X$ in $B$: deleting base edges can create new
induced copies of $H_X$. 

For $F\subseteq E(B)$, let $\Gamma_X^F(B,k)$ denote the graph obtained
from $\Gamma_X(B,k)$ by subdividing each edge of $F$ exactly once, using
one new vertex per edge, and performing no other subdivisions. Original
vertices retain their names. In particular,
\[
                  \Gamma_X^F(B,k)[V(B)]=B-F.
\]

\begin{lemma}[Private-completion lemma]
\label{lem:private-completion}
Let $H$ be fixed, let $X\subseteq V(H)$, and suppose that $H_X=H[X]$
has at least two edges. Assume the following \emph{anchoring property}:
for every graph $B$, every nonnegative integer $k$, every
$F\subseteq E(B)$, and every set $W\subseteq V(\Gamma_X^F(B,k))$ such
that
\[
                         J:=\Gamma_X^F(B,k)[W]
\]
is isomorphic to $H$, there is a graph isomorphism $\psi:H\to J$
satisfying $\psi(X)\subseteq V(B)$.

Then \HFS\ is NP-complete. Furthermore, assuming the ETH, it cannot
be solved in time $2^{o(k)}\cdot n^{O(1)}$.
\end{lemma}

\begin{proof}
We reduce from \pname{$H_X$-free Edge Deletion}, whose hardness follows
from \cref{pro:deletion}. Let $(B,k)$ be a source instance, and put
$n:=|V(B)|$ and $m:=|E(B)|$. If $k\geq m$, deleting every edge makes
$B$ $H_X$-free, so we output the fixed positive instance $(K_1,0)$.
This target instance is positive because $H_X$ has at least two edges
and hence $|V(H)|\geq3$. Otherwise, output $(\Gamma_X(B,k),k)$.

Since $H$ is fixed and $k<m$, the construction can be performed in
polynomial time. 

\emph{Forward direction.}
Suppose that $F\subseteq E(B)$, $|F|\leq k$, and $B-F$ is $H_X$-free.
Subdivide each edge of $F$ once in $\Gamma_X(B,k)$, obtaining
$\Gamma_X^F(B,k)$. If this graph contains an induced copy $J$ of $H$,
the anchoring property gives an isomorphism $\psi:H\to J$ with
$\psi(X)\subseteq V(B)$. Since $J$ is induced, restricting $\psi$ to
$X$ gives an isomorphism
\[
             H_X\cong (B-F)[\psi(X)],
\]
a contradiction. These $|F|\leq k$ subdivisions therefore produce an
$H$-free graph.

\emph{Reverse direction.}
Suppose that a sequence of at most $k$ subdivisions transforms
$\Gamma_X(B,k)$ into an $H$-free graph $G^*$. Let
$T\subseteq E(\Gamma_X(B,k))$ be the set of touched original edges.
Every subdivision acts on the replacement path of a unique original
edge, so $|T|\leq k$, even if edges are subdivided repeatedly. On the
original vertex set, subdivisions remove precisely the touched edges:
\[
        G^*[V(\Gamma_X(B,k))]=\Gamma_X(B,k)-T.
\]
Set $F:=T\cap E(B)$. Then $|F|\leq k$ and $G^*[V(B)]=B-F$.

Suppose for a contradiction that $B-F$ contains an induced copy of
$H_X$. Choose an injection $\varphi\in\mathcal I_X(B)$ that is an
isomorphism from $H_X$ to $(B-F)[\varphi(X)]$. For this fixed
$\varphi$, the $k+1$ sets
$A_{\varphi,1},\ldots,A_{\varphi,k+1}$ are pairwise disjoint.
Consequently, at most $k$ of them meet $T$, and there is an index $i$
with $A_{\varphi,i}\cap T=\emptyset$. All internal and attachment
edges of this private copy therefore survive in $G^*$.

Define a map $\eta:V(H)\to V(G^*)$ by
\[
 \eta(v):=
 \begin{cases}
   \varphi(v),&v\in X,\\
   v_{\varphi,i},&v\in Y.
 \end{cases}
\]
This map is injective. On $X$, it preserves both adjacency and
nonadjacency because $\varphi$ is an induced embedding in $B-F$.
Within $Y$ and between $X$ and $Y$, the definition of the construction
and $A_{\varphi,i}\cap T=\emptyset$ ensure that exactly the edges
prescribed by $H$ survive. No subdivision adds an edge between original
vertices. Hence $\eta$ is an isomorphism from $H$ to
$G^*[\varphi(X)\cup V(Q_{\varphi,i})]$, a contradiction.
Thus $B-F$ is $H_X$-free.

Now the statement follows from \cref{pro:deletion}.
\end{proof}

\begin{lemma}
\label{lem:degree-3}
Let $H$ be a graph, let $X$ be the set of vertices of $H$ having degree at
least $3$, and put $H'=H[X]$.  If $H'$ has at least two edges, then \HFS\ is
NP-complete.  Furthermore, assuming the ETH, it cannot be solved in time
$2^{o(k)}\cdot n^{O(1)}$.
\end{lemma}

\begin{proof}
Apply \cref{lem:private-completion} with this set $X$.  Every vertex outside
the base in the completion construction has degree at most $2$: private
vertices represent vertices of $H-X$, and subdivision vertices have degree
$2$.  Hence every vertex representing a member of $X$ in an induced copy of
$H$ lies in the base.  This is precisely the required anchoring property.
\end{proof}

\begin{proof}[Proof of \cref{thm:easy}]
The first case follows from \cref{lem:edge-deletion-eq,pro:deletion},
and the second from \cref{lem:degree-3}.
\end{proof}

\begin{corollary}
\label{cor:two-bistar-components}
If $H$ has at least two subdivided-bistar components, then
\HFS\ is NP-complete. Furthermore, assuming ETH, it admits no
$2^{o(k)}n^{O(1)}$-time algorithm.
\end{corollary}

\begin{proof}
Each subdivided-bistar component contributes its unsubdivided central edge to
the graph induced by the vertices of degree at least $3$.  Two such components
therefore give two distinct edges, and \cref{lem:degree-3} applies.
\end{proof}

}

\ifthenelse{\boolean{shortver}}{%
\section{Graphs containing a triangle}
\label{sec:girth3}

\begin{theorem}
\label{thm:girth3}
If $H$ is a fixed graph containing a triangle, then \HFS\ is
NP-complete. Furthermore, assuming ETH, it admits no
$2^{o(k)}n^{O(1)}$-time algorithm.
\end{theorem}

Let $D:=\{v\in V(H):\deg_H(v)\geq3\}$. If $H[D]$ has at least
two edges, \cref{lem:degree-3} applies. We outline two reductions
covering the remaining cases; their full proofs appear in the attached
full version.

\begin{restatable}{lemma}{TriangleEdgeAnchorLemma}
\label{lem:triangle-edge-anchor}
Suppose that $H[D]$ has exactly one edge $ab$ and at least one of $a,b$
lies on a triangle. Then \HFS\ is NP-complete and admits no
$2^{o(k)}n^{O(1)}$-time algorithm under ETH.
\end{restatable}

\begin{proof}[Proof sketch]
Let $X$ consist of $a,b$ and all vertices on triangles meeting
$\{a,b\}$, and put $Y:=V(H)\setminus X$. Every vertex of
$X\setminus\{a,b\}$ has degree two and no neighbour in $Y$.
Reduce from \pname{$H[X]$-free Edge Deletion}, using
\cref{pro:deletion}. Given $(B,k)$, assume $k<|E(B)|$, as larger
budgets are trivially positive. For each ordered pair $(u,v)$ with
$uv\in E(B)$, add $k+1$ private copies of $H[Y]$, attached to $u,v$
as $Y$ is attached to $a,b$ in $H$. Add no other edges and retain budget $k$.

If $B-F$ is $H[X]$-free, subdivide each edge of $F$ once.
Every edge with both endpoints of degree at least three lies in the
base: private branching vertices are pairwise nonadjacent and have no
base neighbours. Thus an induced $H$ must represent $a,b$ in the base.
Moreover, a triangle meeting both the base and a private copy would
correspond to a triangle of $H$ containing a vertex of $Y$ and meeting
$\{a,b\}$, contradicting the definition of $X$. Subdivision vertices
lie on no triangle. Hence all vertices representing $X$ lie in the
base, contradicting the choice of $F$.

Conversely, let $F$ be the original base edges touched by a solution.
A surviving induced $H[X]$ determines an ordered base edge representing
$ab$. One of its $k+1$ private copies has no touched internal or
attachment edge, so it completes that $H[X]$ to an induced $H$.
Thus $B-F$ is $H[X]$-free. The construction is polynomial and
preserves $k$, yielding both hardness statements.
\end{proof}

A \emph{windmill} $W_t=K_1\vee(tK_2)$ consists of $t\geq1$
otherwise disjoint triangles sharing their \emph{centre}; $W_1=K_3$.

\begin{restatable}{lemma}{SeparatedWindmillsLemma}
\label{lem:separated-windmills}
Let $X$ be the set of vertices of $H$ lying on triangles. Suppose that
$H[X]$ is a disjoint union of windmills $W_{t_1},\ldots,W_{t_r}$,
where $r,t_1,\ldots,t_r\geq1$. For each component
$X_i:=V(W_{t_i})$, suppose that a centre $a_i$ can be chosen such that
$N_H(y)\cap X_i\subseteq\{a_i\}$ for every $y\in V(H)\setminus X$.
Then \HFS\ is NP-complete and admits no
$2^{o(k)}n^{O(1)}$-time algorithm under ETH.
\end{restatable}

\begin{proof}[Proof sketch]
Relabel so that $s:=t_1=\min_i t_i$, and put $Y:=V(H)\setminus X$.
Reduce from \pname{$K_3$-free Edge Deletion}, which has both hardness
properties by \cref{pro:deletion}. Given $(B,k)$ with $k<|E(B)|$,
set $L:=k+1$ and construct the following graph with budget $k$.

The \emph{active part} $A$ starts with $B$. For every $u\in V(B)$,
add $s-1$ pairs $(U_{u,j},V_{u,j})$ of independent sets of size $L$.
Make each pair complete to each other and to $u$, with no further
edges incident with these sets. For each $i\geq2$, a \emph{protected
component} $P_i$ replaces each vertex $x\in X_i$ by an independent
bag $Z_{i,x}$ of size $L$, with two bags complete precisely when
their corresponding vertices are adjacent. Write $C_i:=Z_{i,a_i}$.
The parts $A,P_2,\ldots,P_r$ are pairwise anticomplete.

For each tuple $(u,c_2,\ldots,c_r)\in V(B)\times C_2\times\cdots
\times C_r$, add $L$ private copies of $H[Y]$, attached to the tuple
as $Y$ is attached to $a_1,\ldots,a_r$ in $H$, with $a_1$ represented
by $u$. Add no other edges. When $r=1$ the tuple consists only of $u$;
when $Y=\emptyset$ no private vertices are needed.

Every triangle lies in an original part $A$ or $P_i$: the attachment
vertices of a private copy are independent, and a triangle using two
private vertices would contradict the definition of $Y$.
This remains true after subdivision. If $B-F$ is triangle-free,
subdividing $F$ once makes the active part $W_s$-free. Indeed, a
windmill centred at $u\in V(B)$ can use at most one triangle from each
of its $s-1$ pairs, while triangles centred at a bag vertex all share
$u$ as well. The case $s=1$ is immediate. Each unchanged $P_i$ is
$2K_3$-free, because every triangle's centre-bag vertex is adjacent to
the non-centre vertices of every other triangle. Thus the $r$ windmill
components of an induced $H[X]$ could occupy neither the active part
nor fewer than $r$ protected components, but only $r-1$ exist.

Conversely, after at most $k$ subdivisions, every bag has a vertex
incident with no touched original edge: each touched edge meets an
independent bag in at most one vertex, and its size is $k+1$.
Let $F$ be the touched base edges. A triangle $uvw$ in $B-F$, together
with clean representatives of the $s-1$ pairs at $u$, induces $W_s$.
Clean representatives in the protected bags provide the other windmills.
For the resulting tuple of centres, one of the $k+1$ private copies
has no touched internal or attachment edge. It completes these
windmills to an induced $H$, a contradiction.
The construction has polynomial size for fixed $H$ and preserves $k$.
\end{proof}

\begin{proof}[Proof of \cref{thm:girth3}]
Apply \cref{lem:degree-3} if $H[D]$ has at least two edges, and
\cref{lem:triangle-edge-anchor} if it has one edge with an endpoint
on a triangle. Otherwise every triangle has at most one vertex in $D$.
A triangle with none is an isolated component; in every other triangle,
the two nonbranching vertices have no neighbours outside it.
Triangles sharing a branching vertex therefore form a windmill.
Distinct centres cannot be adjacent, since their edge would fall in
one of the preceding cases. Hence the triangle vertices induce
separated windmills whose external edges meet only their centres,
and \cref{lem:separated-windmills} applies.
\end{proof}

}{%
\section{Graphs containing a triangle}
\label{sec:girth3}
In this section we prove the following theorem.
\begin{theorem}
\label{thm:girth3}
If $H$ is a fixed graph containing a triangle, then \HFS\ is
NP-complete. Furthermore, assuming ETH, it admits no
$2^{o(k)}n^{O(1)}$-time algorithm.
\end{theorem}

Let $D:=\{v\in V(H):\deg_H(v)\geq3\}$. The case in which $H[D]$
has at least two edges is already covered by \cref{lem:degree-3}.
The next two reductions cover all remaining possibilities.

\begin{restatable}{lemma}{TriangleEdgeAnchorLemma}
\label{lem:triangle-edge-anchor}
Suppose that $H[D]$ has exactly one edge $ab$ and at least one of $a,b$
lies on a triangle. Then \HFS\ is NP-complete and admits no
$2^{o(k)}n^{O(1)}$-time algorithm under ETH.
\end{restatable}
\begin{proof}
Put $A:=\{a,b\}$ and let $X$ consist of $a,b$ and every vertex on a
triangle meeting $A$. Set $H_X:=H[X]$ and $Y:=V(H)\setminus X$.
We use $H_X$ as the deletion pattern, keeping the edge $ab$ and all
triangles meeting $A$. 

Since $H[D]$ has only the edge $ab$, every vertex of $X\setminus A$ has
degree exactly $2$ in $H$, with both neighbours in $X$.
In particular, there are no edges between $X\setminus A$ and $Y$.
Moreover, a vertex of $Y$ of degree at least $3$ has no neighbour in $A$,
and no two such vertices are adjacent. These observations are the anchoring
properties needed below.

We reduce from \pname{$H_X$-free Edge Deletion}. This source problem is
NP-complete and has the required ETH lower bound by \cref{pro:deletion},
because $H_X$ contains a triangle. Given a source instance $(B,k)$, we may
assume $0\leq k<|E(B)|$: otherwise the source instance is trivially positive
and we output $(K_1,0)$.
Start the target graph $G$ with $B$, which we call the base.
For every \emph{ordered} pair $(u,v)$ with $uv\in E(B)$, add $k+1$
private copies of $H[Y]$.
In each such copy, join the copy of $y\in Y$ to $u$ precisely when
$ay\in E(H)$, and to $v$ precisely when $by\in E(H)$.
Add no other edges. Including both orders of each edge handles any
asymmetry between the roles of $a$ and $b$.
The target budget is $k$; the $k+1$ private copies protect the remainder
against this budget, as verified below.

Suppose first that $F\subseteq E(B)$, $|F|\leq k$, and $B-F$ is
$H_X$-free. Subdivide every edge of $F$ once in $G$, obtaining $G_F$.
Every private vertex has exactly its degree in $H$ before these
subdivisions, and this degree is unchanged. Every new subdivision vertex
has degree $2$. Consequently, every edge of $G_F$ whose two endpoints
both have degree at least $3$ lies entirely in $V(B)$: private vertices
of degree at least $3$ are pairwise nonadjacent and have no base neighbours.
If $J$ were an induced copy of $H$, the vertices representing $a,b$
would therefore both belong to $V(B)$.

We claim that every triangle of $G_F$ meeting $V(B)$ lies entirely in
$V(B)$. No subdivision vertex lies on a triangle, and private vertices
from different copies are nonadjacent.
Consider a private copy of $H[Y]$ added for an ordered pair $(u,v)$,
and denote the copy of each $y\in Y$ by $y'$.
A triangle containing $y'$ and two base vertices must use $u,v$, so
$ay,by\in E(H)$. Together with $ab\in E(H)$, this makes $a,b,y$ a
triangle of $H$, contrary to $y\notin X$.
A triangle containing two vertices $y',z'$ of this copy and one base
vertex must use either $u$ or $v$. By construction, $a,y,z$ or $b,y,z$,
respectively, would then form a triangle of $H$, again contradicting
$y,z\notin X$. This proves the claim.
Every vertex of $X\setminus A$ lies on a triangle meeting $A$.
Since the vertices representing $a,b$ in $J$ lie in $V(B)$, the claim
forces all vertices representing $X$ into $V(B)$. Hence $J$ contains
an induced $H_X$ in $G_F[V(B)]=B-F$, a contradiction.

Conversely, suppose that a sequence of at most $k$ subdivisions makes
$G$ $H$-free, and let $F$ be the set of touched edges originally in $B$.
Then $|F|\leq k$, and the original base vertices induce $B-F$ in the final
graph. If $B-F$ contained an induced $H_X$, choose an isomorphism
$\varphi$ from $H_X$ onto such a copy and put
$(u,v):=(\varphi(a),\varphi(b))$.
The edge $uv$ survives, so it was present in $B$ and its two orientations
were included in the construction. The internal and attachment edge sets
of the $k+1$ private copies for $(u,v)$ are pairwise disjoint.
At most $k$ of them can contain a touched original edge. An untouched copy,
together with $\varphi(X)$, still induces $H$, because
$X\setminus A$ has no neighbours in $Y$. This contradiction proves that
$B-F$ is $H_X$-free.

The construction has polynomial size, preserves the parameter, and works
for arbitrary subdivision sequences. This proves both hardness statements.
\end{proof}

A \emph{windmill} $W_t$, for $t\geq1$, consists of
$t$ triangles sharing a single vertex, called the \emph{centre}, and
otherwise vertex-disjoint, with no additional edges. Thus
$W_t=K_1\vee(tK_2)$ and $W_1=K_3$.

\begin{restatable}{lemma}{SeparatedWindmillsLemma}
\label{lem:separated-windmills}
Let $X$ be the set of vertices of $H$ lying on triangles. Suppose that
\[
            H[X]=W_{t_1}\mathbin{\dot\cup}\cdots
                           \mathbin{\dot\cup}W_{t_r},
            \qquad r,t_1,\ldots,t_r\geq1,
\]
and for each component $X_i:=V(W_{t_i})$ there is a choice of centre
$a_i$ such that
\[
        N_H(y)\cap X_i\subseteq\{a_i\}
        \qquad\text{for every }y\in V(H)\setminus X.
\]
Then \HFS\ is NP-complete and admits no
$2^{o(k)}n^{O(1)}$-time algorithm under ETH.
\end{restatable}

\begin{proof}
Write $X_i:=V(W_{t_i})$, let $a_i$ be its specified centre, and relabel the
components so that
\[
                s:=t_1=\min_{1\leq i\leq r}t_i.
\]
Put $Y:=V(H)\setminus X$. The hypothesis means that every edge from $Y$
to $X_i$ ends at $a_i$. The centres $a_1,\ldots,a_r$ are pairwise
nonadjacent, since the $X_i$ are distinct components of $H[X]$.

We reduce from \pname{$K_3$-free Edge Deletion}, which is NP-complete and
has no $2^{o(k)}n^{O(1)}$-time algorithm under ETH by
\cref{pro:deletion}. Let $(B,k)$ be a source instance. As above, we may
assume $0\leq k<|E(B)|$. Set $L:=k+1$.

\paragraph*{The active part.}
Start with a copy of $B$. For every $u\in V(B)$ and every
$j\in\{1,\ldots,s-1\}$, add two independent sets
\[
        U_{u,j},\ V_{u,j},\qquad |U_{u,j}|=|V_{u,j}|=L.
\]
Make these two sets complete to each other, and make $u$ adjacent to all
their vertices. There are no other edges incident with their vertices.
In particular, different pairs of sets are anticomplete.
Call the resulting graph $A$.
If $s=1$, no sets are added and $A=B$.

\paragraph*{The protected components.}
For every $i\in\{2,\ldots,r\}$, form a graph $P_i$ by replacing every
vertex $x\in X_i$ by an independent set $Z_{i,x}$ of size $L$, and making
$Z_{i,x}$ complete to $Z_{i,y}$ exactly when $xy\in E(H[X_i])$.
Let $C_i:=Z_{i,a_i}$ be the centre bag.
The graphs $A,P_2,\ldots,P_r$ are pairwise anticomplete.
Write
\[
             K:=A\mathbin{\dot\cup}P_2\mathbin{\dot\cup}\cdots
                     \mathbin{\dot\cup}P_r.
\]
Each triangle of $P_i$ uses one vertex of $C_i$ and one vertex from
each of the two bags corresponding to a single petal of $W_{t_i}$.
It follows that $P_i$ has no induced $2K_3$: for any two vertex-disjoint
triangles, the centre-bag vertex of the first is adjacent to both
non-centre vertices of the second.

\paragraph*{Private copies of the triangle-free remainder.}
For every tuple
\[
        (u,c_2,\ldots,c_r)\in
            V(B)\times C_2\times\cdots\times C_r,
\]
add $L$ private copies of $H[Y]$.
In each copy, join the vertex representing $y\in Y$ to $u$ if and only
if $ya_1\in E(H)$, and to $c_i$ if and only if $ya_i\in E(H)$ for
$i\geq2$. Add no other edges. For $r=1$, a tuple consists just of $u$.
Call the final graph $G$ and set the target budget equal to $k$ (see \cref{fig:windmill}).
An empty remainder requires no private vertices.

\begin{figure}[t]
    \centering
    \begin{tikzpicture}[every node/.style={draw,circle,fill=white,inner sep=0pt,minimum size=8pt},every loop/.style={}]
\node[label=right:\footnotesize{$u$}] (u1) at (0,0) {};
\node[minimum size=30pt] (w1) at (-2,1) {$U_{u,1}$};
\node[minimum size=30pt] (w2) at (-1,2) {$V_{u,1}$};
\node[minimum size=30pt] (w3) at (1,2) {$U_{u,2}$};
\node[minimum size=30pt] (w4) at (2,1) {$V_{u,2}$};
\node (u2) at (6,0) {};
\node[minimum size=30pt] (w5) at (4,1) {};
\node[minimum size=30pt] (w6) at (5,2) {};
\node[minimum size=30pt] (w7) at (7,2) {};
\node[minimum size=30pt] (w8) at (8,1) {};
\draw[thick] (u1)--(w1)--(w2)--(u1);
\draw[thick] (u1)--(w3)--(w4)--(u1);
\draw[thick] (u2)--(w5)--(w6)--(u2);
\draw[thick] (u2)--(w7)--(w8)--(u2);
\draw[dashed] (-1,0.3)--(7,0.3)--(7,-0.3)--(-1,-0.3)--(-1,0.3);
\draw[fill=black] (3,0) circle (0.03);
\draw[fill=black] (2.7,0) circle (0.03);
\draw[fill=black] (3.3,0) circle (0.03);
\node[draw=none,fill=none] (b) at (-1.5,0) {$B$};

\node[minimum size=30pt] (c) at (0,-3) {$C_2$};
\node[minimum size=30pt] (c1) at (-1,-1.5) {};
\node[minimum size=30pt] (c2) at (-2,-2.5) {};
\node[minimum size=30pt] (c3) at (1,-1.5) {};
\node[minimum size=30pt] (c4) at (2,-2.5) {};
\node[minimum size=30pt] (c5) at (-0.7,-4.5) {};
\node[minimum size=30pt] (c6) at (0.7,-4.5) {};
\draw[thick] (c)--(c1)--(c2)--(c);
\draw[thick] (c)--(c3)--(c4)--(c);
\draw[thick] (c)--(c5)--(c6)--(c);

\draw[fill=black] (3.5,-3) circle (0.03);
\draw[fill=black] (3.2,-3) circle (0.03);
\draw[fill=black] (3.8,-3) circle (0.03);

\node[minimum size=30pt] (cr) at (7,-3) {$C_r$};
\node[minimum size=30pt] (b1) at (6,-1.5) {};
\node[minimum size=30pt] (b2) at (5,-2.5) {};
\node[minimum size=30pt] (b3) at (8,-1.5) {};
\node[minimum size=30pt] (b4) at (9,-2.5) {};

\node[minimum size=30pt] (b5) at (6,-5) {};
\node[minimum size=30pt] (b6) at (5,-4) {};
\node[minimum size=30pt] (b7) at (8,-5) {};
\node[minimum size=30pt] (b8) at (9,-4) {};

\draw[thick] (cr)--(b1)--(b2)--(cr);
\draw[thick] (cr)--(b3)--(b4)--(cr);
\draw[thick] (cr)--(b5)--(b6)--(cr);
\draw[thick] (cr)--(b7)--(b8)--(cr);

\end{tikzpicture}
    \caption{Part of the construction in \cref{lem:separated-windmills}. Big circle denotes an independent set, and line between two sets denotes that there are all possible edges. Additionally to the vertices drawn in the pictures, there are copies of $H[Y]$, namely $L$ copies for each tuple $(u,c_2,\ldots,c_r)$, where $u$ is a vertex of $B$ and $c_i\in C_i$, and those copies are connected to $\{u,c_2,\ldots,c_r\}$ according to edges of $H$.}
    \label{fig:windmill}
\end{figure}

We establish the following three claims.

\begin{claim}
\label{clm:windmills-triangles}
After any sequence of subdivisions of $G$, every triangle has all its
vertices in $V(K)$.
\end{claim}
\begin{proof}[Proof of the claim]
A private vertex has base neighbours only among the mutually nonadjacent
vertices $u,c_2,\ldots,c_r$. Hence it cannot form a triangle with two
vertices of $K$.
A triangle with two private vertices and one vertex of $K$ would have to
use vertices of one private copy. Its corresponding vertices in $H$ would
form a triangle containing a vertex of $Y$, contrary to the definition of
$X$. A triangle entirely in a private copy is excluded for the same reason.
Distinct private copies are anticomplete. Thus every triangle of $G$ lies
in $K$. This remains true after any subdivisions, since subdividing an edge
never creates a triangle.
\end{proof}

\begin{claim}
\label{clm:windmills-active}
Let $F\subseteq E(B)$ be such that $B-F$ is triangle-free, and let $A_F$
be obtained from $A$ by subdividing once every edge of $F$.
Then $A_F$ is $W_s$-free and hence is $W_{t_i}$-free for every $i$.
\end{claim}
\begin{proof}[Proof of the claim]
For $s=1$ the claim follows because $A_F$ is triangle-free.
Suppose $s\geq2$. Every triangle of $A_F$ consists of an original vertex
$u\in V(B)$, one vertex of $U_{u,j}$, and one of $V_{u,j}$ for some $j$.
If the centre of an induced windmill is $u$, it can use at most one
triangle from each pair $(U_{u,j},V_{u,j})$: two petals from the same
pair would have unwanted cross-edges from the complete bipartite graph.
Only $s-1$ pairs are available. If its centre lies in one of these bags,
every triangle through that centre also contains $u$, so two triangles
cannot intersect only at the proposed centre. A subdivision vertex lies
on no triangle. In all cases an induced $W_s$ is impossible.
\end{proof}

\begin{claim}
\label{clm:windmills-untouched}
For every sequence of at most $k$ subdivisions of $G$, each bag
$U_{u,j}$, $V_{u,j}$, or $Z_{i,x}$ contains a vertex incident with no
touched original edge.
\end{claim}
\begin{proof}[Proof of the claim]
Consider an arbitrary sequence of at most $k$ subdivisions of $G$.
Call an original vertex \emph{dirty} if it is an endpoint of a touched
original edge. Every bag $U_{u,j},V_{u,j},Z_{i,x}$ is independent in $G$.
Thus each touched edge has at most one endpoint in a specified bag.
There are at most $k$ touched original edges, so at most $k$ vertices of
any bag are dirty. Since every bag has size $k+1$, it contains a vertex
incident with no touched original edge. All original edges incident with
such a vertex survive in the final graph.
This argument also covers subdivisions of attachment edges and repeated
subdivisions of descendants.
\end{proof}

\paragraph*{Forward implication.}
Suppose that $F\subseteq E(B)$, $|F|\leq k$, and $B-F$ is triangle-free.
Subdivide every edge of $F$ once in $G$ and call the result $G_F$.
Assume for contradiction that $G_F$ contains an induced copy $J$ of $H$.
By \cref{clm:windmills-triangles}, the image of $X$ is contained in the
original vertices of $K$. Since every $H[X_i]=W_{t_i}$ is connected and the parts
$A,P_2,\ldots,P_r$ are pairwise anticomplete, the image of each $X_i$
lies wholly in the active part or wholly in one protected component.
\Cref{clm:windmills-active} excludes the active part.
Furthermore, one protected component cannot contain images of two distinct
$X_i$: one triangle from each image would give an induced $2K_3$, whereas
each $P_j$ is $2K_3$-free.
Hence the $r$ components of $H[X]$ would have to occupy $r$ distinct
members of $\{P_2,\ldots,P_r\}$, a set of size $r-1$.
This contradiction proves that $G_F$ is $H$-free.

\paragraph*{Reverse implication.}
Suppose that a sequence of at most $k$ subdivisions makes $G$ $H$-free.
Let $F$ be the set of touched edges originally in $B$; then $|F|\leq k$.
If $B-F$ contains a triangle $uvw$, we will construct an induced copy of
$H$ that survives the sequence.
Represent the centre $a_1$ by $u$, and use $v,w$ for the two vertices
of one petal of $W_s$. For each of the other $s-1$ petals, choose an
untouched representative from each of $U_{u,j},V_{u,j}$, using
\cref{clm:windmills-untouched}. These vertices and $u,v,w$ still induce
$W_s$ in the final graph: the source triangle survived, the selected
petal edges are untouched, and no cross-edges were added.
For every $i\geq2$, use \cref{clm:windmills-untouched} to choose an
untouched representative from every bag $Z_{i,x}$.
These representatives still induce $H[X_i]=W_{t_i}$.
Together, the selected vertices induce $H[X]$. Let $c_i$ denote the
selected representative of $a_i$ for each $i\geq2$.

For the tuple $(u,c_2,\ldots,c_r)$ the construction supplied $k+1$
private copies of $H[Y]$. Their internal and attachment edge sets are
pairwise disjoint, so at least one such copy has no touched edge.
Select it. Its prescribed adjacencies to the chosen centres, the absence
of edges to the petal vertices, and the absence of edges between distinct
base parts show that the selected vertices induce exactly $H$ in the
final graph. This contradicts $H$-freeness.
Thus $B-F$ is triangle-free, completing the equivalence.

For fixed $H$, the construction is polynomial. More explicitly, it uses
$nL^{r-1}$ tuples and at most $nL^r|Y|$ private vertices; its total
number of edges is
$O_H(m+(n+1)L^2+nL^r)$. Since we assumed $k<m$, these bounds are polynomial
in the source input size. The parameter is unchanged, so both NP-hardness
and the stated ETH lower bound transfer from
\pname{$K_3$-free Edge Deletion}.
\end{proof}

\begin{proof}[Proof of \cref{thm:girth3}]
If $H[D]$ has at least two edges, apply \cref{lem:degree-3}.
If it has exactly one edge with an endpoint lying on a triangle, apply
\cref{lem:triangle-edge-anchor}.
Otherwise every triangle has at most one vertex in $D$.
A triangle with no vertex in $D$ is an isolated $K_3$ component of $H$.
In every other triangle, its two vertices outside $D$ have degree exactly
$2$ and have no neighbours outside the triangle. All triangles sharing the
same vertex of $D$ therefore form a windmill, with all external neighbours
incident only with its centre. Distinct such centres cannot be adjacent,
because their edge would have an endpoint on a triangle and would have
been covered by one of the first two cases.
Consequently the graph induced by all triangle vertices is a disjoint union
of windmills satisfying \cref{lem:separated-windmills}, which applies.
These cases are exhaustive and each gives the ETH lower bound.
\end{proof}

}

\ifthenelse{\boolean{shortver}}{%
\section{Triangle-free graphs containing a cycle}
\label{sec:girth4}

\begin{theorem}
\label{thm:girth-4}
Let $H$ be a fixed graph containing a cycle and no triangle.
Then \HFS\ is NP-complete and, assuming ETH, admits no
$2^{o(k)}n^{O(1)}$-time algorithm.
\end{theorem}

\begin{proof}
Put $h:=|V(H)|$ and $Q:=\core(H)$. Thus $Q$ is nonempty,
triangle-free, and has minimum degree at least two.
By \cref{lem:degree-3}, we may assume that the vertices of degree
at least three in $H$ induce at most one edge. Consequently, every
cycle of $Q$ contains a vertex of degree two in $Q$.
Let $U$ consist of the vertices on shortest cycles of $Q$.
We use the following structural fact; the attached full version gives
the complete proof.

\begin{claim}
\label{clm:girth4-alt-structure}
At least one of the following holds:
\begin{enumerate}
\item $\overline Q$ is $3$-connected;
\item $Q[U]$ is $2$-connected;
\item $Q$ consists of $K_{2,r}$ and $K_{2,s}$, for some $r,s\geq2$,
sharing exactly one vertex from their respective sides of size two,
with no other edges.
\end{enumerate}
\end{claim}

\begin{proof}[Proof sketch of the claim]
If (1) fails, there is a set $S$ of at most two vertices such that
$\overline Q-S$ is disconnected. Triangle-freeness forces
$Q-S$ to be complete bipartite with nonempty sides $A,B$.
If both sides have size at least two, their vertices lie on $4$-cycles.
Adjoining $S\cap U$ to this $2$-connected graph preserves
$2$-connectivity: add a vertex with two neighbours already present,
or a two-vertex path with distinct ends already present.

Otherwise, let $A=\{x\}$. Minimum degree two implies that each
vertex of $B$ has a neighbour in $S$. If $|S|=1$ or $x$ has a
neighbour in $S$, then $Q$ is a complete bipartite graph $K_{2,r}$.
The remaining case has $S=\{y,z\}$, with $x$ adjacent to neither.
Put $B_t=N_Q(t)\cap B$ and
$Q_t:=Q[\{x,t\}\cup B_t]\cong K_{2,|B_t|}$ for
$t\in\{y,z\}$; the sets $B_y,B_z$ cover $B$.
If $yz$ is an edge, they are disjoint and nonempty: either $Q=C_5$,
or its shortest cycles lie in the subgraphs $Q_t$ with
$|B_t|\geq2$. Their vertices induce a $2$-connected graph, using
$yz$ when both subgraphs contribute. If $yz$ is absent, both
subgraphs are $2$-connected; intersecting $B_y,B_z$ give (2),
and disjoint $B_y,B_z$ give (3).
\end{proof}

\paragraph*{Attaching the trees outside the core.}
Every component of $H-V(Q)$ is a tree with at most one edge to $Q$.
For $x\in V(Q)$, let $T_x$ contain $x$ and the trees attached to it,
rooted at $x$; let $R$ consist of the tree components of $H$ disjoint
from $Q$. Given a base $B$ and budget $K$, form $P_H(B,K)$ by
attaching $K+1$ private copies of every nontrivial $T_x$ at every
vertex of $B$, identifying roots, and adding $K+1$ disjoint copies
of $R$.

After at most $K$ arbitrary subdivisions, let $F$ be the touched
original base edges. Every induced $Q$ in $B-F$ extends to $H$:
choose an untouched copy of each required rooted tree and of $R$.
Such copies exist because their internal and attachment edge sets
are disjoint. Conversely, subdivisions confined to the base that
make it $Q$-free make the whole graph $H$-free. Indeed, repeatedly
deleting vertices of degree at most one removes every vertex outside
the subdivided base,
so every copy of the minimum-degree-two graph $Q$ lies in that base.
An induced $H$ would contain an induced $Q$.

\paragraph*{Case 1: $\overline Q$ is $3$-connected.}
Apply \cref{pro:deletion-gap} with $c=h$ to a $3$-SAT instance
with $N$ variables and clauses in total, obtaining $(B,p)$, and
output $(P_H(B,hp),hp)$.
For a satisfiable instance, delete at most $p$ edges to make $B$
$Q$-free, and instead subdivide each of these edges $h$ times.
A copy of $Q$ containing an internal subdivision vertex would,
by minimum degree two, contain its entire replacement path,
which has $h+2>|V(Q)|$ vertices. Thus the subdivided base is
$Q$-free, and the attachment construction gives an $H$-free outcome.
Conversely, an $H$-free outcome within budget $hp$ gives a
$Q$-free deletion set of size at most $hp$ in $B$, contradicting
the gap guarantee for an unsatisfiable instance. Since $p=O(N)$,
the target budget is $O(N)$, giving the ETH lower bound.

\paragraph*{The remaining constructions.}
Take a \pname{Vertex Cover} instance $(G,k)$ with
$\girth(G)>h$ and $k\leq |V(G)|$, using \cref{pro:vc}.
Start the base with an independent set $V(G)$ and a vertex $q$
adjacent to it. Each source edge $uv$ receives an \emph{edge piece}
with nonadjacent attachment vertices $u,v$; we may also attach a
\emph{vertex piece} at each source vertex. All other vertices are
private, and distinct pieces intersect only at their attachments.
Add no other edges.

Two facts control these constructions. After subdividing any of
the edges $qv$ once, every $2$-connected subgraph of order at most
$h$ avoiding $q$ lies in one piece. Indeed, a cycle using several
pieces projects to a cycle of $G$, of length greater than $h$,
and any two edges of a $2$-connected graph lie on a common cycle.
The new subdivision vertices are leaves after deleting $q$.
Also, suppose every path within an edge piece between its attachments
has length at least $g-2$, where $g\geq4$.
Subdividing $qv$ once for each vertex $v$ of a cover leaves no
$g$-cycle through $q$: a cycle traversing one edge piece has length
at least $(g-2)+3$, and one traversing at least two has length
at least $2(g-2)+2>g$.

\paragraph*{Case 2: $Q[U]$ is $2$-connected.}
Put $g:=\girth(Q)$, and choose a degree-two vertex $\tau$ on a
$g$-cycle, with neighbours $a,b$. For each arbitrarily oriented
source edge $(u,v)$, its edge piece is a copy of $Q-\tau$ with
$a,b$ identified with $u,v$. Together with $q$, it induces $Q$.
Output $(P_H(B,k),k)$.

Given a cover $C$ of size at most $k$, subdivide every $qv$ with
$v\in C$. Every $a$--$b$ path in $Q-\tau$ has length at least
$g-2$, so $q$ lies on no $g$-cycle. In any surviving induced
$Q$, the image of $U$ therefore avoids $q$, and its $2$-connected
subgraph lies in one edge piece. This is impossible: $Q-\tau$
has strictly fewer $g$-cycles than $Q$, while $Q[U]$ contains
all of them. The base is thus $Q$-free, giving an $H$-free outcome.
Conversely, from any successful subdivision sequence, select $v$
for every touched $qv$, and one endpoint of each source edge
whose piece contains a touched original edge. Charge these
selections to distinct touched base edges. An uncovered edge
would leave its piece together with $q$ inducing $Q$ on original
vertices, which extends to $H$. Thus the selected vertices form
a cover of size at most $k$.

\paragraph*{Case 3: the two complete bipartite bundles.}
Write $Q=K_{2,r}\cup K_{2,s}$ as in the claim, with $r\leq s$.
The edge piece for $uv$ is $K_{2,r-1}$ with $u,v$ on its side
of size two. The vertex piece at $v$ is $K_{2,s}$ with $v$ on
its side of size two. Output $(P_H(B,k),k)$.
For every source edge $uv$, its edge piece, $q$, and the vertex
piece at $v$ induce $Q$, with shared vertex $v$.

A cover of size at most $k$ yields an $H$-free outcome by
subdividing its edges $qv$.
Indeed, $q$ then lies on no $4$-cycle, so both bundles of a
surviving $Q$ avoid $q$ and each lies in one piece. Neither fits
in an edge piece, which has only $r+1$ vertices. Distinct vertex
pieces are disjoint, whereas the bundles intersect; one vertex
piece cannot contain their union either, since all its induced
subgraphs are complete bipartite and $Q$ is not.
Conversely, select $v$ for a touched $qv$ or a touched edge in
its vertex piece, and one endpoint of each source edge with a
touched edge piece. These selections again cost at most $k$.
An uncovered edge leaves the displayed induced $Q$ intact,
contradicting the extension property.

All constructions have polynomial size. Cases 2 and 3 preserve
$k$, so \cref{pro:vc}, together with $k\leq |V(G)|$, gives the
claimed ETH lower bound. All reverse arguments allow repeated
subdivisions and subdivisions outside the base.
\end{proof}

}{%
\section{Triangle-free graphs containing a cycle}
\label{sec:girth4}
In this section we prove hardness for all triangle-free graphs that contain a cycle, completing the dichotomy for all graphs that are not forests.
\begin{theorem}
\label{thm:girth-4}
Let $H$ be a fixed graph containing a cycle and no triangle.
Then \HFS\ is NP-complete. Furthermore, assuming ETH,
it admits no $2^{o(k)}n^{O(1)}$-time algorithm.
\end{theorem}

\begin{proof}
Put $h:=|V(H)|$ and $Q:=\core(H)$. Thus $Q$ is nonempty,
triangle-free, and has minimum degree at least $2$.
If the vertices of degree at least $3$ in $H$ induce at least two
edges, apply \cref{lem:degree-3}. Henceforth assume that they induce
at most one edge. In particular, every cycle of $Q$ contains a
vertex of degree $2$ in $Q$.

\paragraph*{An elementary structural observation.}
Let $U$ be the set of vertices on shortest cycles of $Q$.
\begin{claim}
\label{clm:girth4-alt-structure}
At least one of the following holds:
\begin{enumerate}
\item $\overline Q$ is $3$-connected;
\item $Q[U]$ is $2$-connected;
\item $Q$ consists of $K_{2,r}$ and $K_{2,s}$, for some $r,s\geq2$,
sharing exactly one vertex from their respective sides of size $2$,
with no other edges.
\end{enumerate}
\end{claim}
\begin{proof}[Proof of the claim]
Suppose that $\overline Q$ is not $3$-connected. Since $Q$ is nonempty,
triangle-free, and has minimum degree at least $2$, it has at least
four vertices. Hence, by the definition of $3$-connectivity, there exists
$S\subseteq V(Q)$ with $|S|\leq2$ such that $\overline Q-S$ is disconnected.
Here $S$ may be empty if $\overline Q$ is already disconnected, or a
singleton if it has a cut vertex. Such an $S$ can be enlarged to size
two while retaining vertices in two distinct components, but we allow
all three sizes below.
Triangle-freeness implies that $Q-S$ is a complete bipartite graph
with nonempty independent sides $A,B$: take one component of
$\overline Q-S$ and its complement, and note that an edge within
either side would form a triangle with every vertex of the other side.

If $|A|,|B|\geq2$, then $\girth(Q)=4$ and $A\cup B\subseteq U$.
The graph $Q[A\cup B]$ is $2$-connected. A vertex of $S$ with at least
two neighbours in $A\cup B$ lies on a $4$-cycle and can be added
while preserving $2$-connectivity. Any remaining vertex
$s\in S\cap U$ has exactly one neighbour in $A\cup B$ and is adjacent
to the other vertex $t$ of $S$, since $\delta(Q)\geq2$ and $|S|\leq2$.
Consequently, $\deg_Q(s)=2$, so every cycle containing $s$ also
contains $t$. As $s\in U$ lies on a $4$-cycle, that cycle also
contains $t$, proving that $t\in U$.
If $t$ has at least two neighbours in $A\cup B$, add $t$ first and
then $s$, preserving $2$-connectivity at each step.
Otherwise, $s$ and $t$ each have exactly one neighbour in $A\cup B$,
and these neighbours are distinct by triangle-freeness.
They therefore form a path with distinct ends in $A\cup B$;
adding this path again preserves $2$-connectivity.
Thus $Q[U]$ is $2$-connected.

Otherwise, after interchanging $A$ and $B$ if necessary, let
$A=\{x\}$. Every vertex of $B$ has $x$ as its only neighbour
outside $S$, so $\delta(Q)\geq2$ forces it to have a neighbour
in $S$. In particular, $S\neq\emptyset$.
If $S=\{y\}$, then $y$ is adjacent to every vertex of $B$, and
$xy\notin E(Q)$ by triangle-freeness. Thus
$Q\cong K_{2,|B|}$ with sides $\{x,y\}$ and $B$, where
$|B|\geq2$ follows from $\deg_Q(x)\geq2$.
Consequently, $U=V(Q)$ and $Q[U]$ is $2$-connected.

Now suppose $S=\{y,z\}$. If $x$ has a neighbour in $S$, say $y$,
then triangle-freeness makes $y$ nonadjacent to every vertex of $B$.
Hence $\deg_Q(y)\geq2$ forces $yz\in E(Q)$, which in turn forces
$xz\notin E(Q)$. Every vertex of $B$ must therefore be adjacent
to $z$. Thus $Q\cong K_{2,|B|+1}$ with sides $\{x,z\}$ and
$B\cup\{y\}$. Since $B\neq\emptyset$, again $U=V(Q)$ and
$Q[U]$ is $2$-connected.

It remains that $x$ is nonadjacent to both $y,z$. Put
\[
B_y:=N_Q(y)\cap B,\qquad B_z:=N_Q(z)\cap B.
\]
Since every vertex of $B$ has a neighbour in $S$, we have
$B_y\cup B_z=B$. For $t\in\{y,z\}$, let
$Q_t:=Q[\{x,t\}\cup B_t]\cong K_{2,|B_t|}$.
All edges of $Q$ belong to $Q_y$ or $Q_z$, except possibly $yz$.

Suppose first that $yz\in E(Q)$. Triangle-freeness gives
$B_y\cap B_z=\emptyset$, and minimum degree at least $2$ makes
both sets nonempty. Thus $Q_y$ and $Q_z$ meet exactly at $x$.
Every cycle avoiding $yz$ lies in one of these two subgraphs,
whereas every cycle using $yz$ has length at least $5$.
If $|B_y|=|B_z|=1$, then $Q$ is a $5$-cycle, so $U=V(Q)$.
Otherwise, $\girth(Q)=4$. If exactly one of $B_y,B_z$ has size
at least two, say $B_y$, then $U=V(Q_y)$, and hence $Q[U]=Q_y$
is $2$-connected. If both have size at least two, then $U=V(Q)$
and both $Q_y,Q_z$ are $2$-connected. In this case $Q-x$ is
connected through $yz$, while deleting any other vertex leaves
$Q_y$ and $Q_z$ connected and still meeting at $x$.
Thus $Q[U]$ is $2$-connected in every case with $yz\in E(Q)$,
establishing condition~(2).

Finally, suppose $yz\notin E(Q)$. Minimum degree at least $2$
gives $|B_y|,|B_z|\geq2$, so both $Q_y,Q_z$ are $2$-connected,
$Q=Q_y\cup Q_z$, and $U=V(Q)$.
If $B_y\cap B_z\neq\emptyset$, then $Q_y,Q_z$ share at least
two vertices. After deleting any vertex, they remain connected
and still intersect; hence $Q[U]=Q$ is $2$-connected.
If $B_y\cap B_z=\emptyset$, the two complete bipartite subgraphs
share exactly $x$, which belongs to their respective sides of
size $2$, and there are no other edges. This is condition~(3).
\end{proof}

\paragraph*{Attaching the part outside the core.}
In Case~1, we reduce from \pname{3-SAT} using
\cref{pro:deletion-gap}. In Cases~2 and~3, we reduce from
\pname{Vertex Cover} on graphs of girth greater than $h$, using
\cref{pro:vc}. Each reduction will specify a base graph $B$ and
a target subdivision budget $K$, and output the instance
$(P_H(B,K),K)$. 
Before giving the case-specific reductions, we describe the common
construction $P_H(B,K)$ for an arbitrary graph $B$ and a
nonnegative integer $K$.

Every component of $H-V(Q)$ is a tree with at most one edge to $Q$.
For $x\in V(Q)$, let $T_x$ be the rooted tree consisting of $x$ and
all such components attached to $x$, with root $x$. Let $R$ be the
union of the tree components of $H$ disjoint from $Q$.
Construct $P_H(B,K)$ by attaching $K+1$ private copies of every
nontrivial $T_x$ at every vertex of $B$, identifying their roots
with that vertex, and adding $K+1$ disjoint copies of $R$.
Add no other edges.
The following two claims establish the extension and preservation
properties of this construction.

\begin{claim}[Extension property]
\label{clm:girth4-extension}
Let $P'$ be obtained from $P_H(B,K)$ by at most $K$ arbitrary
subdivisions, and let $F\subseteq E(B)$ be the set of touched
original edges of $B$. Every induced copy of $Q$ in $B-F$ extends
to an induced copy of $H$ in $P'$. Consequently, if $P'$ is
$H$-free, then $B-F$ is $Q$-free.
\end{claim}
\begin{proof}[Proof of the claim]
The original base vertices induce $B-F$ in $P'$. Let
$\varphi\colon V(Q)\to V(B)$ be an induced embedding of $Q$ into
$B-F$. For each $x\in V(Q)$ with $T_x$ nontrivial, consider the
$K+1$ copies of $T_x$ attached at $\varphi(x)$. Their edge sets
are pairwise disjoint, so at most $K$ of them contain a touched
original edge. Choose one untouched copy. Likewise, if $R$ is
nonempty, choose an untouched copy among its $K+1$ copies.
These choices remain possible when subdivisions are repeated,
since each operation acts on a unique original edge or one of
its descendants. The selected copies retain all their original
edges. The selected tree copies meet the base only at their
prescribed roots, and the selected copy of $R$, when needed, is
disjoint from the base. Distinct selected copies have disjoint
private vertex sets with no edges between them. Together with
$\varphi(V(Q))$, they therefore induce $H$ in $P'$.
\end{proof}

\begin{claim}[Preservation property]
\label{clm:girth4-preservation}
Any sequence of subdivisions confined to original edges of $B$
and their descendants that makes the base $Q$-free also makes
$P_H(B,K)$ $H$-free.
\end{claim}
\begin{proof}[Proof of the claim]
Let $B'$ and $P'$ be the resulting subdivisions of $B$ and
$P_H(B,K)$, respectively, and suppose that $B'$ is $Q$-free.
All vertices outside $B'$ can be removed by successively deleting
vertices of degree at most one: delete the non-root vertices of
each attached tree in decreasing order of distance from its root,
and remove the copies of $R$ by deleting leaves and isolated
vertices. No vertex of a subgraph of minimum degree at
least $2$ can be removed in this process, since the first such
vertex removed would still have at least two neighbours in that
subgraph. As $\delta(Q)\geq2$, every induced copy of $Q$ in $P'$
would therefore lie entirely in $B'$, a contradiction. Finally,
an induced copy of $H$ would contain an induced copy of
$Q=\core(H)$. Thus $P'$ is $H$-free.
\end{proof}

\paragraph*{Case 1: $\overline Q$ is $3$-connected.}
Apply \cref{pro:deletion-gap} with $c=h$ to a $3$-SAT instance of size $N$,
obtaining $(B,p)$, and output $(P_H(B,hp),hp)$.
If the source instance is satisfiable, there exists $F\subseteq E(B)$
with $|F|\leq p$ such that $B-F$ is $Q$-free.
Subdivide each edge of $F$ exactly $h$ times. An induced copy of $Q$
cannot contain an internal vertex of one of the resulting paths:
minimum degree at least $2$ would force it to contain the entire
path, which has $h+2>|V(Q)|$ vertices. Hence the subdivided base is
$Q$-free, and \cref{clm:girth4-preservation} gives an $H$-free outcome
within budget $hp$. Conversely, an $H$-free outcome within this
budget yields, by \cref{clm:girth4-extension}, a $Q$-free edge-deletion
set in $B$ of size at most $hp$, which is impossible for an unsatisfiable
source instance by \cref{pro:deletion-gap}. Since $p=O(N)$, the target
budget is $O(N)$, so this reduction also gives the claimed ETH lower bound.

\paragraph*{Two observations for the remaining constructions.}
For the remaining cases, take a \pname{Vertex Cover} instance $(G,k)$
with $\girth(G)>h$ and $0\leq k\leq |V(G)|$, using \cref{pro:vc}.
We construct a base graph $B$ starting with an independent set
$V(G)$ and a new vertex $q$ adjacent to every vertex of $V(G)$.
For each edge $uv\in E(G)$, attach one graph with two distinguished
nonadjacent vertices identified with $u$ and $v$; the resulting
copy, including $u$ and $v$, is called the \emph{edge piece} for $uv$.
We may also attach, at each $v\in V(G)$, at most one graph with a
distinguished vertex identified with $v$; the resulting copy,
including $v$, is called the \emph{vertex piece} at $v$.
The vertices of a piece identified with vertices of $G$ are its
\emph{source vertices}. All other vertices of each piece are new
and \emph{private} to that piece, so distinct pieces intersect
only at source vertices.
The graphs used as pieces are specified in Cases~2 and~3 below.
The base $B$ has no edges other than those within the pieces and
the edges $qv$ for $v\in V(G)$.

\begin{claim}
\label{clm:girth4-piece-containment}
After subdividing any subset of the edges $qv$ once, every
$2$-connected subgraph of the resulting base of order at most $h$
avoiding $q$ lies in one piece.
\end{claim}
\begin{proof}[Proof of the claim]
The new subdivision vertices are leaves after deleting $q$.
If a cycle avoiding $q$ uses edges from several pieces, each
maximal subpath within one piece joins two distinct source
vertices, and hence lies in an edge piece. Replacing these
subpaths by their corresponding source edges gives a cycle of
$G$: the source vertices do not repeat, and a cycle using only
two source vertices would lie entirely in their unique edge
piece. The projected cycle has order at most that of the original
cycle, so the latter has order greater than $h$.
Since any two edges of a $2$-connected graph lie on a common
cycle, a $2$-connected subgraph of order at most $h$ avoiding $q$
cannot use edges from different pieces.
\end{proof}

\begin{claim}
\label{clm:girth4-no-g-cycle}
Suppose that every path within an edge piece between its two
source vertices has length at least $g-2$, where $g\geq4$.
If $C$ is a vertex cover of $G$, then subdividing $qv$ once for
each $v\in C$ yields a base in which $q$ lies on no $g$-cycle.
\end{claim}
\begin{proof}[Proof of the claim]
Any cycle through $q$ consists of two paths of length one or two
from $q$ to distinct source vertices, together with a path between
them. As in the proof of \cref{clm:girth4-piece-containment},
this last path traverses $t\geq1$ edge pieces, each between its
two source vertices.
If $t=1$, its ends form an edge of $G$, so at least one belongs
to $C$. Consequently, at least one of the two paths from $q$
has length two, giving a cycle length of at least $(g-2)+3=g+1$.
If $t\geq2$, the cycle has length at least $2(g-2)+2>g$.
Thus $q$ lies on no $g$-cycle.
\end{proof}

\paragraph*{Case 2: $Q[U]$ is $2$-connected.}
Put $g:=\girth(Q)$ and choose a degree-$2$ vertex $\tau$ on a
$g$-cycle of $Q$; write $N_Q(\tau)=\{a,b\}$.
Such a vertex exists by our initial reduction to
\cref{lem:degree-3}. Fix an arbitrary orientation of $G$.
For each oriented edge $(u,v)$, attach a private copy of
$Q-\tau$ with $a,b$ identified with $u,v$, respectively.
Together with $q$, this piece induces $Q$, with $q$ representing
$\tau$. These are all the pieces of the base $B$.
Output $(P_H(B,k),k)$.

Let $C$ be a vertex cover with $|C|\leq k$ and subdivide $qv$
once for every $v\in C$, obtaining a base $B_C$.
Every $a$--$b$ path in $Q-\tau$ has length at least $g-2$,
so \cref{clm:girth4-no-g-cycle} implies that $q$ lies on no
$g$-cycle of $B_C$. Suppose an induced copy of $Q$ survives.
Since every vertex of $U$ lies on a $g$-cycle, its image avoids
$q$. The image of $Q[U]$ is $2$-connected and has order at most
$h$, so it lies in one edge piece by
\cref{clm:girth4-piece-containment}.
That piece is isomorphic to $Q-\tau$, which has strictly fewer
$g$-cycles than $Q$. It cannot contain the images of all the
$g$-cycles of $Q$, a contradiction. Thus $B_C$ is $Q$-free,
and \cref{clm:girth4-preservation} gives an $H$-free outcome.

Conversely, consider at most $k$ subdivisions giving an $H$-free
outcome. For each touched edge $qv$, select $v$. For each edge
piece containing a touched original edge, select one endpoint of
its source edge. These selections can be charged to distinct
touched base edges, so at most $k$ vertices are selected.
If a source edge $uv$ is uncovered, its edge piece together with
$q$ still induces $Q$ on original base vertices, contradicting
\cref{clm:girth4-extension}. The selected vertices are therefore a
vertex cover.

\paragraph*{Case 3: two complete bipartite bundles.}
Write $Q$ as $K_{2,r}\cup K_{2,s}$ with just one shared vertex
$x$ on their sides of size $2$, where $2\leq r\leq s$.
For every edge $uv\in E(G)$, add $r-1$ private vertices adjacent
exactly to $u,v$; these form an edge piece $K_{2,r-1}$.
For every $v\in V(G)$, add a private vertex $b_v$ and $s$ private
common neighbours of $v,b_v$. These form a vertex piece
$K_{2,s}$. Let $B$ be the resulting base and output
$(P_H(B,k),k)$. For every source edge $uv$, its edge piece together
with $q$ induces $K_{2,r}$, and adding the vertex piece at $v$
gives an induced $Q$ with shared vertex $v$.

Let $C$ be a vertex cover of size at most $k$ and subdivide $qv$
once for every $v\in C$, obtaining a base $B_C$.
Every path within an edge piece between its source vertices has
length at least $2$, so \cref{clm:girth4-no-g-cycle}, with $g=4$,
implies that $q$ lies on no $4$-cycle.
In a surviving induced copy of $Q$, both bundles would therefore
avoid $q$, since every vertex of either bundle lies on a $4$-cycle.
Each bundle is $2$-connected and has order at most $h$, so
\cref{clm:girth4-piece-containment} places it in one piece.
Neither fits in an edge piece, which has only $r+1$ vertices.
Thus both lie in vertex pieces. Distinct vertex pieces are
vertex-disjoint, so they cannot accommodate two bundles sharing
$x$. Nor can one vertex piece contain both bundles: its induced
subgraphs are complete bipartite, whereas $Q$ is not.
Thus $B_C$ is $Q$-free, and \cref{clm:girth4-preservation}
again applies.

For the reverse implication, select $v$ for each touched edge
$qv$, select one endpoint of $uv$ if its edge piece contains a
touched original edge, and select $v$ if its vertex piece contains
a touched original edge. The three kinds of edge sets are pairwise
disjoint, so at most $k$ vertices are selected. An uncovered source
edge $uv$ would leave the induced $Q$ described above untouched,
contradicting \cref{clm:girth4-extension}. Thus the selected vertices
form a vertex cover.

The cases are exhaustive by \cref{clm:girth4-alt-structure}.
The constructions in Cases~2 and~3 have polynomial size and preserve
$k$, so \cref{pro:vc} gives both hardness statements; the bound
$k\leq |V(G)|$ yields the stated parameter lower bound from its
$2^{o(n+m)}$ lower bound. Case~1 has the same lower bound by
\cref{pro:deletion-gap}, as shown above. All reverse implications allow repeated subdivisions and
subdivisions of attached-tree edges.
\end{proof}

}

\section{Trees}
\label{sec:trees}

We handle trees with exactly two branching vertices; the
case of at least three branching vertices is treated in
\cref{sec:forests}.

\begin{theorem}
\label{thm:tree}
Let $H$ be a fixed tree with exactly two branching vertices at
distance at least two. Then \HFS\ is NP-complete. Furthermore,
assuming ETH, it admits no $2^{o(k)}n^{O(1)}$-time algorithm.
\end{theorem}

A single reduction handles every distance except $3$: it uses
\pname{$P_3$-free Edge Deletion} for even distance and
\pname{$P_4$-free Edge Deletion} for odd distance.  The remaining case is
handled in \cref{subsec:distance-three} by a parameter-preserving reduction
from \pname{Vertex Cover} on triangle-free graphs.

\begin{lemma}
\label{lem:tree-distance-not-three}
Let $H$ be a tree with exactly two branching vertices at distance $d$, where
$d\geq2$ and $d\neq3$.  Then \HFS\ is NP-complete and, assuming the ETH, has
no $2^{o(k)}\cdot n^{O(1)}$-time algorithm.
\end{lemma}

\begin{proof}[Proof of \cref{lem:tree-distance-not-three}]
Set
\[
 s:=\begin{cases}
       2,&\text{if $d$ is even},\\
       3,&\text{if $d$ is odd},
    \end{cases}
 \qquad r:=\frac{d-s}{2}.
\]
Thus $r\geq0$, and the odd case has $d\geq5$.  We reduce from
\pname{$P_{s+1}$-free Edge Deletion}.  Let $(G_0,k_0)$ be an instance.  If
$k_0\geq |E(G_0)|$, the instance is positive and we output $(K_1,0)$.
Hence $k_0<|E(G_0)|$; put $K:=dk_0$.  This also makes the
construction below polynomial in the input size.

Let $P$ be the path between the branching vertices of $H$.  Every component
of $H-V(P)$ is a path attached to one endpoint of $P$.  Let $\lambda$ be the maximum length of a pendant arm from its branching
vertex to its leaf, \emph{including} the edge attaching the component to
$P$; in particular $\lambda\geq1$. Put
\[
        c:=\max\{\deg_H(u)-1,\deg_H(v)-1\},
\]
where $u,v$ are the branching vertices.  For $q\in\{0,\ldots,r\}$, let $A_q$
be the rooted tree obtained from a root-to-centre path of length $q$ by
attaching $c$ paths of length $\lambda$ to its centre.  At every vertex of
$G_0$, attach $K+1$ otherwise disjoint copies of every $A_q$.

Suppose first that $F\subseteq E(G_0)$, $|F|\leq k_0$, and $G_0-F$ is
$P_{s+1}$-free.  Subdivide every edge of $F$ exactly $d$ times.  If the
resulting graph contained $H$, let $\alpha,\beta$ be its branching vertices
and let $q_\alpha,q_\beta\in\{0,\ldots,r\}$ be their distances from the roots
of their attached gadgets; an original vertex is treated as a centre with $q=0$. All vertices that
can branch are original vertices or gadget centres. Two such vertices at the
same root have distance at most $2r<d$, so their roots must be distinct.
The part of the branch-to-branch path that lies in the subdivided
copy of $G_0$ has length
\[
        p=d-q_\alpha-q_\beta\geq d-2r=s.
\]
It cannot traverse an edge of $F$, whose replacement has length $d+1$.
Consequently it yields an induced path of length at least $s$ in $G_0-F$,
contradicting that this graph is $P_{s+1}$-free.

Conversely, suppose at most $K$ subdivisions produce an $H$-free graph.  For
an original edge $e\in E(G_0)$, let $t(e)$ be the total number of subdivisions
performed on $e$ and its descendants, and put
\[
        F:=\{e\in E(G_0):t(e)\geq d\}.
\]
Then $|F|\leq k_0$.  If $G_0-F$ contained an induced path
$v_0v_1\ldots v_s$, set
\[
        p_i:=t(v_{i-1}v_i)+1\in[1,d]
        \qquad(i\in[s]).
\]
Some nonempty consecutive sum of the $p_i$ lies in $[s,d]$.  For $s=2$, take
one $p_i\geq2$, if it exists, and otherwise take $p_1+p_2=2$.  For $s=3$,
take one $p_i\geq3$, if it exists; otherwise all $p_i\leq2$, and either their
total is $3$ or a term equal to $2$ together with an adjacent term has sum
$3$ or $4\leq d$.

The corresponding section of the subdivided path is an induced path $Q$ of
length $p\in[s,d]$ between two original vertices $a,b$.  Since
$d-p\leq d-s=2r$, choose $q_a,q_b\in\{0,\ldots,r\}$ with
$q_a+q_b=d-p$.  At most $K$ attached copies are touched, so untouched copies
of $A_{q_a}$ at $a$ and $A_{q_b}$ at $b$ exist.  Their centres, the path $Q$,
and suitable prefixes of their pendant paths induce a copy of $H$, a
contradiction.  Hence $G_0-F$ is $P_{s+1}$-free.

The parameter is multiplied only by the constant $d$, and the result follows
from \cref{pro:deletion}.
\end{proof}

Rooted pendant-arm
gadgets absorb the difference between a source path length and the required
branch distance. The case $d=3$ requires a different construction.

\subsection{The case in which the branching vertices are at distance three}
\label{subsec:distance-three}

A set $M$ of vertices is a \emph{module} if every vertex outside $M$ is either
complete or anticomplete to $M$.

\begin{theorem}
\label{thm:two-branch-distance-three}
Let $H$ be a tree with exactly two branching vertices at distance $3$.  Then
\HFS\ is NP-complete and, assuming the ETH, has no
$2^{o(k)}\cdot n^{O(1)}$-time algorithm.
\end{theorem}

\begin{proof}
We reduce from \pname{Vertex Cover} on triangle-free graphs.  Let $(Q,p)$ be
such an instance.  Write the path between the branching vertices of $H$ as
\[
        b_1z_1z_2b_2.
\]
For $i\in\{1,2\}$, let
$\ell_{i,1},\ldots,\ell_{i,q_i}$ be the lengths of the pendant paths attached
to $b_i$ away from this path.  Put
\[
 q:=\max\{q_1,q_2\},\qquad
 L:=\max\{\ell_{i,j}:i\in\{1,2\},\ j\in[q_i]\}.
\]

For every $x\in V(Q)$, take a subdivided star $M_x$ with centre $c_x$, having
$q$ arms of length $L$, and add a vertex $a_x$ adjacent to $c_x$.  Make
$A:=\{a_x:x\in V(Q)\}$ a clique.  For distinct $x,y$, make $M_x$ complete to
$M_y$ when $xy\notin E(Q)$ and anticomplete to $M_y$ when $xy\in E(Q)$.
The output is $(G,p)$.

For every edge $xy\in E(Q)$, the induced path
$c_xa_xa_yc_y$ extends, using suitable prefixes of distinct arms in $M_x$ and
$M_y$, to an induced copy of $H$.

Suppose first that $C$ is a vertex cover of size at most $p$.  Subdivide
$c_xa_x$ once for every $x\in C$, and call the resulting graph $G_C$.  Let
\[
        M:=G_C\Big[\bigcup_x V(M_x)\Big].
\]
We use the following simple fact: every induced $P_5$ in $M$ lies in one
module $M_x$.  Indeed, $P_5$ has no nontrivial proper module.  Thus, if an
induced $P_5$ met several sets $M_x$, it would use one vertex from each of five
distinct modules, giving an induced $P_5$ in $\overline Q$.  The first, third,
and fifth vertices would then form a triangle in $Q$.

Assume that $G_C$ contains an induced copy $T$ of $H$, and let
$ux_1x_2v$ be its branch-to-branch path.  Neither a vertex of $A$ nor a
subdivision vertex can be branching: the neighbours selected from the clique
$A$ are pairwise adjacent, and every selector has only one neighbour outside
$A$.

If the branch-to-branch path meets $A$, then both internal vertices lie in
$A$, and the path is $c_xa_xa_yc_y$ for some $x,y\notin C$.  Since it is
induced, $M_x$ and $M_y$ are anticomplete, so $xy\in E(Q)$, contradicting that
$C$ covers $Q$.

Otherwise the path lies in $M$. Every vertex of $M$ has at most
one neighbour outside $M$, so at least one of the two pendant
neighbours of $u$ in $T$ also lies in $M$. Together with the
branch-to-branch path it gives an
induced $P_5$ in $M$.  Hence both $u$ and $v$ lie in one module $M_t$.  In an
induced tree, the only vertex of $M_t$ that can be branching while having a
neighbour in $M_t$ is its centre $c_t$: a noncentre has degree at most $2$
inside $M_t$, and any selected neighbour in another module is adjacent also to
its neighbour in $M_t$, creating a triangle.  Thus $u=v=c_t$, a contradiction.
Therefore $G_C$ is $H$-free.

Conversely, suppose that at most $p$ subdivisions make $G$ $H$-free.  Call an
original edge touched if it or one of its descendants is subdivided.  Put $x$
in $C'$ whenever $c_xa_x$ is touched, and, for every $xy\in E(Q)$ whose edge
$a_xa_y$ is touched, put one arbitrary endpoint of $xy$ in $C'$.  If an edge
$xy\in E(Q)$ were uncovered by $C'$, none of
\[
        c_xa_x,\quad a_xa_y,\quad a_yc_y
\]
would be touched.  The resulting unchanged path can still be completed to
$H$, because subdivisions on the arms only increase their lengths.  This is a
contradiction.  Hence $C'$ is a vertex cover, and each insertion is charged to
a distinct touched edge, so $|C'|\leq p$.

For fixed $H$, the construction has $O(|V(Q)|)$ vertices and
$O(|V(Q)|^2)$ edges; it preserves the parameter. The result follows from
\cref{pro:vc} with girth bound $4$.
\end{proof}

\begin{proof}[Proof of \cref{thm:tree}]
The cases $d\neq3$ and $d=3$ follow respectively from
\cref{lem:tree-distance-not-three,thm:two-branch-distance-three}.
\end{proof}

Thus every tree with exactly two nonadjacent branching vertices is
covered by the hardness results, while the adjacent case is
polynomial-time solvable.

\section{The forest dichotomy}
\label{sec:forests}

For a tree $T$ with at least two branching vertices, its \emph{branch core}
$C(T)$ is the minimal subtree containing all branching vertices.  An edge of
$C(T)$ is called a \emph{branch-core edge}.

\begin{lemma}
\label{lem:subdividing-branch-core}
Let $T$ be a tree with at least two branching vertices, and let $T'$
be a subdivision of $T$ in which at least one edge of $C(T)$ is
replaced by a path of length at least two. Then $T'$ is $T$-free.
\end{lemma}

\begin{proof}
Let $B$ be the set of branching vertices of $T$. Subdivision preserves
the degrees of original vertices and creates only degree-$2$ vertices,
so $B$ is also the set of branching vertices of $T'$.
Suppose that $J$ is an induced copy of $T$ in $T'$. Every branching
vertex of $J$ belongs to $B$, and $J$ has exactly $|B|$ branching
vertices; hence its branching vertices are precisely $B$.
Since $J$ is connected and paths in $T'$ are unique, the minimal
subtree of $J$ containing $B$ is $C(T')$. Thus $C(J)=C(T')$.
However, $C(J)\cong C(T)$, whereas $C(T')$ is a proper subdivision
of $C(T)$ and therefore has more vertices, a contradiction.
\end{proof}

We also use the following elementary observation.

\begin{lemma}
\label{lem:selector-linear-forest}
Let $S$ be obtained from a clique $A=\{a_x:x\in X\}$ by attaching
to each $a_x$ either the edge $a_xb_x$ or the path $a_xs_xb_x$,
with all added vertices distinct. Every induced forest $F$ of $S$
is a linear forest, and at most one component of $F$ contains a
vertex of degree $2$ in $F$.
\end{lemma}

\begin{proof}
Since $F$ is induced and acyclic, it contains at most two vertices
of $A$. Each such vertex has at most one neighbour in $A\cap V(F)$
and at most one outside $A$, while every vertex outside $A$ has
degree at most $2$ in $S$. Thus every vertex of $F$ has degree at
most $2$, so $F$ is a linear forest.
Moreover, $S-A$ is a disjoint union of edges and isolated vertices.
Hence every component of $F$ containing a degree-$2$ vertex meets
$A$. Since $A$ is a clique, at most one component of $F$ meets $A$,
proving the second assertion.
\end{proof}

\begin{lemma}
\label{lem:many-branching-tree}
Let $T$ be a fixed tree with at least three branching vertices. Then
\pname{$T$-free Subdivision} is NP-complete and, assuming the ETH,
admits no $2^{o(k)}\cdot n^{O(1)}$-time algorithm.
\end{lemma}

\begin{proof}
Let $B$ be the set of branching vertices of $T$. If $T[B]$ has at
least two edges, \cref{lem:degree-3} applies. Hence assume
\begin{equation}
\label{eq:few-branch-edges}
        |E(T[B])|\leq 1.
\end{equation}
Every leaf of $C(T)$ belongs to $B$. If $C(T)$ had diameter at most
two, then, since $|B|\geq3$, its centre would belong to $B$ and be
adjacent to at least two leaves in $B$, contradicting
\eqref{eq:few-branch-edges}. Thus $C(T)$ contains a path
$z_0z_1z_2z_3$, whose first and last edges are disjoint branch-core
edges.

We reduce from \pname{Vertex Cover}. Let $(Q,p)$ be an instance with
$p\leq |V(Q)|$, and orient each edge arbitrarily. For every
$x\in V(Q)$, introduce an edge $a_xb_x$, and make
$A:=\{a_x:x\in V(Q)\}$ a clique. Call the vertices $a_x,b_x$
the \emph{original selectors}. For each oriented edge $f=(x,y)$,
take a fresh copy $F_f$ of $T$ and identify $z_0,z_1,z_2,z_3$ with
$b_x,a_x,a_y,b_y$, respectively. The remaining vertices form its
private set $P_f$; these sets are pairwise disjoint.
Finally, make $P_f$ complete to $P_g$ for $f\neq g$, and make
$P_f$ complete to $\{a_z,b_z\}$ for every $z\notin\{x,y\}$.
Add no other edges. Denote the resulting graph by $G$ and output
$(G,p)$. Each $F_f$, called the \emph{$f$-frame}, remains an induced
copy of $T$.

Let $C$ be a vertex cover of size at most $p$. Subdivide $a_xb_x$
once for every $x\in C$, obtaining $G_C$. Let $S_C$ be the subgraph
induced by the original selectors and the subdivision vertices.
For each $f=(x,y)$, write $F_f^C$ for the subdivided $f$-frame and
$D_f$ for its path from $b_x$ to $b_y$ through $a_x,a_y$.
Every vertex of $P_f$ has at most one neighbour in $D_f$, since
$F_f^C$ is a tree.

We show that every induced copy $J$ of $T$ in $G_C$ is contained
in some $F_f^C$. Fix such a copy $J$.

\begin{claim}
\label{clm:many-branching-attachments}
Let $W\subseteq V(J)$ be nonempty, contain all private vertices
of $J$, and satisfy that $J[W]$ is connected. Then at most one
vertex outside $W$ is branching in $J$, and any such vertex
has a neighbour in $W$.
\end{claim}

\begin{proof}
By \cref{lem:selector-linear-forest}, $J-W$ is a linear forest,
with at most one component containing degree-$2$ vertices.
Each component has exactly one edge to $W$, by connectedness
and acyclicity of $J$. A branching vertex outside $W$ must
therefore have degree $2$ in $J-W$ and be incident with that
component's unique edge to $W$. Hence at most one such vertex
exists.
\end{proof}

Put
\[
        \mathcal P:=\{f:V(J)\cap P_f\neq\emptyset\}.
\]
If $\mathcal P=\emptyset$, \cref{lem:selector-linear-forest}
contradicts the existence of branching vertices in $J$.
Since distinct private sets are complete to one another,
acyclicity gives $|\mathcal P|\leq2$. If two occur, one
intersection is a singleton and the other is independent;
otherwise a triangle or a $4$-cycle arises.

Suppose first that $\mathcal P=\{f\}$, and put
$R:=V(J)\cap P_f$. If $|R|=1$, applying
\cref{clm:many-branching-attachments} with $W=R$ gives at most
one branching vertex outside $R$, and hence at most two in $J$,
a contradiction. Thus $|R|\geq2$.
Assume that $J$ is not contained in $F_f^C$.
Since all private vertices of $J$ lie in $P_f$, a vertex of $J$
outside the frame is either an original selector outside $D_f$
or a subdivision vertex whose only neighbours are original
selectors outside $D_f$. In the latter case, connectedness
forces at least one of these neighbours to belong to $J$.
Hence $J$ contains an original selector $q$ outside $D_f$.
By construction, $q$ is complete to $R$, so $R$ is independent;
otherwise an edge of $J[R]$ would form a triangle with $q$.
Moreover, no second original selector $q'$ outside $D_f$ belongs
to $J$: for distinct $r_1,r_2\in R$, the cycle
$qr_1q'r_2q$ would lie in $J$.
Since subdivision vertices have no private neighbours, every
neighbour in $J$ of a vertex $r\in R$ is therefore either $q$
or a vertex of $D_f$. As $F_f^C$ is a tree, $r$ has at most one
neighbour in $D_f$, and hence $\deg_J(r)\leq2$.
Now $W:=R\cup\{q\}$ contains all private vertices of $J$ and
induces a star. Only $q$ can be a branching vertex of $J$ in $W$,
while \cref{clm:many-branching-attachments} gives at most one
branching vertex outside $W$. Thus $J$ has at most two branching
vertices, again a contradiction.

Finally, suppose that $\mathcal P=\{f,g\}$, with
$V(J)\cap P_f=\{q\}$ and $R:=V(J)\cap P_g$ independent.
Put $W:=R\cup\{q\}$, which induces a star.
If $|R|\geq2$, no original selector outside $D_g$ belongs to $J$:
together with $q$ and two vertices of $R$, it would form a
$4$-cycle. Thus every vertex of $R$ has degree at most two,
and \cref{clm:many-branching-attachments} again gives at most
two branching vertices in $J$.
If $R=\{r\}$, the same claim gives at most one branching vertex
outside $\{q,r\}$. Since $J$ has at least three branching vertices,
they must be $q,r$ and one vertex $v$ adjacent to $\{q,r\}$.
The edge $qr$ and the edge from $v$ to $\{q,r\}$ then both join
branching vertices, contradicting \eqref{eq:few-branch-edges}.
This proves the required containment.

Since $C$ covers every source edge, each $F_f^C$ subdivides at
least one of the two chosen branch-core edges. Hence every
$F_f^C$ is $T$-free by \cref{lem:subdividing-branch-core}.
The containment just proved implies that $G_C$ is $T$-free.

Conversely, suppose that at most $p$ subdivisions make $G$
$T$-free. Call an original edge touched if it or one of its
descendants is subdivided, and put
\[
        C_0:=\{x:a_xb_x\text{ is touched}\}.
\]
For $f=(x,y)$, let $I_f$ consist of the edges of $F_f$ other
than $a_xb_x$ and $a_yb_y$. These sets are pairwise disjoint:
the edge $a_xa_y$ is unique to $f$, and every other edge of
$I_f$ is incident with $P_f$.
If $xy$ is not covered by $C_0$, some edge of $I_f$ must be
touched; otherwise the original vertices of $F_f$ still
induce $T$. Add one endpoint of every such uncovered edge
to $C_0$. For each vertex $x$ initially in $C_0$, charge its inclusion to
the touched selector edge $a_xb_x$. Whenever a new vertex is
added to cover an edge $xy$, charge its inclusion to a touched
edge of $I_f$, where $f$ is the chosen orientation of $xy$.
If the selected endpoint is already present, no new charge is
needed. Different additions correspond to different source edges,
so their charges are distinct because the sets $I_f$ are pairwise
disjoint. Moreover, no $I_f$ contains a selector edge, so these
charges are also distinct from those of the initial vertices.
Thus each vertex of the resulting set is charged to a distinct
touched original edge. Since each touched original edge requires
at least one subdivision, this set has size at most $p$.
It is a vertex cover of $Q$, since it contains $C_0$ and an
endpoint of every edge initially uncovered by $C_0$.

For fixed $T$, the construction has polynomial size and
preserves the budget. Together with $p\leq |V(Q)|$,
\cref{pro:vc} therefore gives NP-hardness and the claimed
ETH lower bound.
\end{proof}

\begin{theorem}
\label{thm:connected-tree-dichotomy}
If a connected tree $T$ is neither a subdivided star nor a subdivided bistar,
then \pname{$T$-free Subdivision} is NP-complete and, assuming the ETH, has no
$2^{o(k)}\cdot n^{O(1)}$-time algorithm.
\end{theorem}

\begin{proof}
A tree with at most one branching vertex is a subdivided star.  With exactly
two branching vertices it is a subdivided bistar precisely when they are
adjacent, so \cref{thm:tree} handles the remaining cases.  Trees with at least
three branching vertices are covered by \cref{lem:many-branching-tree}.
\end{proof}

We next lift the connected result to arbitrary forests.

\begin{lemma}
\label{lem:component-padding}
Let $H$ be a fixed forest and let $T$ be a component of $H$ that is neither a
subdivided star nor a subdivided bistar.  Suppose that no component of $H$
contains $T$ as a proper induced subgraph.  A parameter-preserving hardness
result, and hence a $2^{o(k)}\cdot n^{O(1)}$ ETH lower bound, for
\pname{$T$-free Subdivision} transfers to \HFS.
\end{lemma}

\begin{proof}
Let $\mu$ be the number of components of $H$ isomorphic to $T$, let $R$ be
obtained from $H$ by deleting one such component, let $h:=|V(T)|$, and let
$(G,k)$ be an instance of
\pname{$T$-free Subdivision}.  If $k\geq h|E(G)|$, subdividing every edge of
$G$ exactly $h$ times makes the distance between distinct original vertices
at least $h+1$.  Since all branching vertices remain original and $T$ has at
least two of them, the resulting graph is $T$-free.  We output the fixed positive instance $(K_1,0)$ in this case and henceforth
assume $k<h|E(G)|$.

Add $k+1$ disjoint copies $R_1,\ldots,R_{k+1}$ of $R$, make distinct copies
complete to each other, and leave all of them anticomplete to $G$.  This graph
has polynomial size and the parameter remains $k$.

Suppose first that subdivisions inside $G$ make it $T$-free.  A connected
induced forest meeting two pairwise-complete copies $R_i,R_j$ is a star: it
meets at most two copies, and if it meets two, one intersection is a singleton
and the other is independent.  Since $T$ is not a subdivided star, every induced copy of $T$ is contained
in $G$ or in one $R_i$.  The first is impossible.  By hypothesis, the only
components of $R$ containing an induced $T$ are its $\mu-1$ components
isomorphic to $T$.  Moreover, copies of $T$ from different $R_i$ cannot be
distinct components of an induced $H$, because all cross edges are present.
Thus the constructed graph has at most $\mu-1$ pairwise anticomplete induced
$T$-components and is $H$-free.

Conversely, suppose at most $k$ subdivisions make the constructed graph
$H$-free, and retain only those performed on edges originating in $G$.  If the
resulting subdivision of $G$ contained $T$, one of the $k+1$ copies $R_i$
would have all its internal edges untouched.  That copy of $R$, together with
the copy of $T$ in $G$, would induce $H$, a contradiction.  Thus the retained
subdivisions make $G$ $T$-free.
\end{proof}

\begin{theorem}
\label{thm:forest-dichotomy}
Let $H$ be a forest.  If $H$ has at least two subdivided-bistar components, or
has a component that is neither a subdivided star nor a subdivided bistar,
then \HFS\ is NP-complete.  Furthermore, assuming the ETH, it has no
$2^{o(k)}\cdot n^{O(1)}$-time algorithm.
\end{theorem}

\begin{proof}
The first case is \cref{cor:two-bistar-components}.  In the second case, choose
among the components that are neither subdivided stars nor subdivided bistars
a component $T$ maximal under the induced-subgraph relation.  By
\cref{thm:connected-tree-dichotomy}, the problem for $T$ is hard with the
claimed lower bound.  No component of $H$ contains $T$ properly: this follows
from maximality for the hard components, while every connected induced
subgraph of a subdivided star or subdivided bistar is again of one of those
two types.  Apply \cref{lem:component-padding}.
\end{proof}

\ifthenelse{\boolean{shortver}}{%
\begin{proof}[Proof of \cref{thm:dichotomy}]
The polynomial-time side is \cref{obs:poly}. Suppose the stated structural
condition fails. If $H$ contains a triangle, apply \cref{thm:girth3}.
Otherwise, if $H$ is a forest, apply \cref{thm:forest-dichotomy}; if it is
not a forest, its girth is at least four and \cref{thm:girth-4} applies.
The same cases give the ETH lower bound.
\end{proof}
}{%
\DichotomyTheorem*

\begin{proof}
The polynomial-time side is \cref{obs:poly}. Suppose the stated structural
condition fails. If $H$ contains a triangle, apply \cref{thm:girth3}.
Otherwise, if $H$ is a forest, apply \cref{thm:forest-dichotomy}; if it is
not a forest, its girth is at least four and \cref{thm:girth-4} applies.
The same cases give the ETH lower bound.
\end{proof}
}

\section*{Declaration of AI Use}

The hardness results for triangle-free graphs containing
cycles, graphs containing exactly one triangle, several other
triangle-containing configurations, and trees with exactly two
branching vertices at distance $2$ or at least $4$ were obtained
without generative AI assistance.

OpenAI's GPT-6 Pro was used to resolve the
remaining cases and complete the complexity dichotomy, as well
as to simplify arguments and generate initial drafts of proofs.
All AI-assisted material was reviewed and revised by the authors,
who take full responsibility for the results.

\bibliographystyle{plainurl}
\bibliography{main}

\begin{thebibliography}{10}

\bibitem{antony2024switching}
Dhanyamol Antony, Yixin Cao, Sagartanu Pal, and RB~Sandeep.
\newblock Switching classes: Characterization and computation.
\newblock In {\em 49th International Symposium on Mathematical Foundations of Computer Science (MFCS 2024)}, pages 11:1--11:15. Schloss Dagstuhl -- Leibniz-Zentrum f{\"u}r Informatik, 2024.
\newblock \href {https://doi.org/10.4230/LIPIcs.MFCS.2024.11} {\path{doi:10.4230/LIPIcs.MFCS.2024.11}}.

\bibitem{antony2022subgraph}
Dhanyamol Antony, Jay Garchar, Sagartanu Pal, RB~Sandeep, Sagnik Sen, and R~Subashini.
\newblock On subgraph complementation to h-free graphs.
\newblock {\em Algorithmica}, 84(10):2842--2870, 2022.

\bibitem{antony2024cutting}
Dhanyamol Antony, Sagartanu Pal, RB~Sandeep, and R~Subashini.
\newblock Cutting a tree with subgraph complementation is hard, except for some small trees.
\newblock {\em Journal of Graph Theory}, 107(1):126--168, 2024.

\bibitem{DBLP:journals/siamdm/AravindSS17}
N.~R. Aravind, R.~B. Sandeep, and Naveen Sivadasan.
\newblock Dichotomy results on the hardness of h-free edge modification problems.
\newblock {\em {SIAM} J. Discret. Math.}, 31(1):542--561, 2017.
\newblock \href {https://doi.org/10.1137/16M1055797} {\path{doi:10.1137/16M1055797}}.

\bibitem{BliznetsCKP18}
Ivan Bliznets, Marek Cygan, Pawe{\l} Komosa, and Micha{\l} Pilipczuk.
\newblock Hardness of approximation for {$H$}-free edge modification problems.
\newblock {\em ACM Transactions on Computation Theory}, 10(2):9:1--9:32, 2018.
\newblock \href {https://doi.org/10.1145/3196834} {\path{doi:10.1145/3196834}}.

\bibitem{chakraborty2023contracting}
Dipayan Chakraborty and RB~Sandeep.
\newblock Contracting edges to destroy a pattern: A complexity study.
\newblock In {\em International Symposium on Fundamentals of Computation Theory}, pages 118--131. Springer, 2023.
\newblock \href {https://doi.org/10.1007/978-3-031-43587-4_9} {\path{doi:10.1007/978-3-031-43587-4_9}}.

\bibitem{cygan2015parameterized}
Marek Cygan, Fedor~V Fomin, {\L}ukasz Kowalik, Daniel Lokshtanov, D{\'a}niel Marx, Marcin Pilipczuk, Micha{\l} Pilipczuk, and Saket Saurabh.
\newblock {\em Parameterized algorithms}.
\newblock Springer, 2015.
\newblock \href {https://doi.org/10.1007/978-3-319-21275-3} {\path{doi:10.1007/978-3-319-21275-3}}.

\bibitem{firbas2024complexity}
Alexander Firbas and Manuel Sorge.
\newblock On the complexity of establishing hereditary graph properties via vertex splitting.
\newblock In {\em 35th International Symposium on Algorithms and Computation (ISAAC 2024)}, volume 322 of {\em Leibniz International Proceedings in Informatics (LIPIcs)}, pages 30:1--30:15. Schloss Dagstuhl -- Leibniz-Zentrum f{\"u}r Informatik, 2024.
\newblock \href {https://doi.org/10.4230/LIPIcs.ISAAC.2024.30} {\path{doi:10.4230/LIPIcs.ISAAC.2024.30}}.

\bibitem{jana2026hardnessvertexsplittingcographs}
Satyabrata Jana, Shivesh~K. Roy, and R.~B. Sandeep.
\newblock Hardness of vertex splitting: Cographs, chordal graphs, and beyond, 2026.
\newblock URL: \url{https://arxiv.org/abs/2607.13517}, \href {https://arxiv.org/abs/2607.13517} {\path{arXiv:2607.13517}}.

\bibitem{johnson2025complexity}
Matthew Johnson, Barnaby Martin, Jelle~J Oostveen, Sukanya Pandey, Dani{\"e}l Paulusma, Siani Smith, and Erik~Jan van Leeuwen.
\newblock Complexity framework for forbidden subgraphs i: The framework.
\newblock {\em Algorithmica}, 87(3):429--464, 2025.

\bibitem{komusiewicz2018tight}
Christian Komusiewicz.
\newblock Tight running time lower bounds for vertex deletion problems.
\newblock {\em ACM Transactions on Computation Theory (TOCT)}, 10(2):1--18, 2018.

\bibitem{lewis1980node}
John~M. Lewis and Mihalis Yannakakis.
\newblock The node-deletion problem for hereditary properties is {NP}-complete.
\newblock {\em Journal of Computer and System Sciences}, 20(2):219--230, 1980.
\newblock \href {https://doi.org/10.1016/0022-0000(80)90060-4} {\path{doi:10.1016/0022-0000(80)90060-4}}.

\bibitem{poljak1974note}
Svatopluk Poljak.
\newblock A note on stable sets and colorings of graphs.
\newblock {\em Commentationes Mathematicae Universitatis Carolinae}, 15(2):307--309, 1974.

\bibitem{robertson1986graph}
Neil Robertson and Paul~D Seymour.
\newblock Graph minors. v. excluding a planar graph.
\newblock {\em Journal of Combinatorial Theory, Series B}, 41(1):92--114, 1986.

\bibitem{DBLP:journals/siamcomp/Yannakakis81}
Mihalis Yannakakis.
\newblock Edge-deletion problems.
\newblock {\em {SIAM} J. Comput.}, 10(2):297--309, 1981.
\newblock \href {https://doi.org/10.1137/0210021} {\path{doi:10.1137/0210021}}.

\end{thebibliography}

\end{document}